\documentclass[twocolumn,english,superscriptaddress,floatfix,longbibliography]{revtex4-2}
\usepackage{mypackage}

\begin{document}

\title{Nontrivial multi-product commutation relation toward reducing $T$-count\\ in sequential Pauli-based computation}

\author{Yusei Mori}
\email{u839526i@ecs.osaka-u.ac.jp}
\affiliation{Graduate School of Engineering Science, The University of Osaka, 1-3 Machikaneyama, Toyonaka, Osaka 560-8531, Japan}

\author{Hideaki Hakoshima}
\email{hakoshima.hideaki.es@osaka-u.ac.jp}
\affiliation{Graduate School of Engineering Science, The University of Osaka, 1-3 Machikaneyama, Toyonaka, Osaka 560-8531, Japan}
\affiliation{Center for Quantum Information and Quantum Biology, The University of Osaka, 1-2 Machikaneyama, Toyonaka, Osaka 560-0043, Japan}

\author{Keisuke Fujii}
\email{fujii.keisuke.es@osaka-u.ac.jp}
\affiliation{Graduate School of Engineering Science, The University of Osaka, 1-3 Machikaneyama, Toyonaka, Osaka 560-8531, Japan}
\affiliation{Center for Quantum Information and Quantum Biology, The University of Osaka, 1-2 Machikaneyama, Toyonaka, Osaka 560-0043, Japan}
\affiliation{RIKEN Center for Quantum Computing, 2-1 Hirosawa, Wako, Saitama 351-0198, Japan}

\begin{abstract}
  Quantum compilers that reduce the number of $T$~gates are essential for minimizing the overhead of fault-tolerant quantum computation.
  Achieving further $T$-count reduction calls for identifying equivalent circuit transformation rules beyond those utilized in existing tools.
  In this paper, we rewrite any given Clifford$+T$ circuit using a Clifford block followed by a sequential Pauli-based computation, and introduce a nontrivial, ancilla-free transformation rule, the multi-product commutation relation (MCR).
  MCR constructs gate sequences based on specific commutation properties among multi-Pauli operators, yielding seemingly non-commutative instances that can be commuted, thereby enabling gate orderings that cannot be derived from pairwise commutation alone.
  We also propose the MCR Compiler, which incorporates MCR-based transformations as an optimization pass. To evaluate its effect, we use a benchmark circuit dataset generated through quantum circuit unoptimization. This approach intentionally adds redundancy to the circuit while keeping its equivalence, allowing a quantitative evaluation of compiler performance by comparison with the original circuit.
  Our numerical experiments reveal that the MCR Compiler achieves further $T$-count reduction beyond current compilers, establishing MCR-based transformations as a practical optimization primitive. These results highlight an untapped opportunity to enhance the optimization capabilities of quantum compilers.
\end{abstract}

\maketitle

\section{Introduction}

\begin{figure*}[t]
  \centering
  \includegraphics[width=0.95\linewidth]{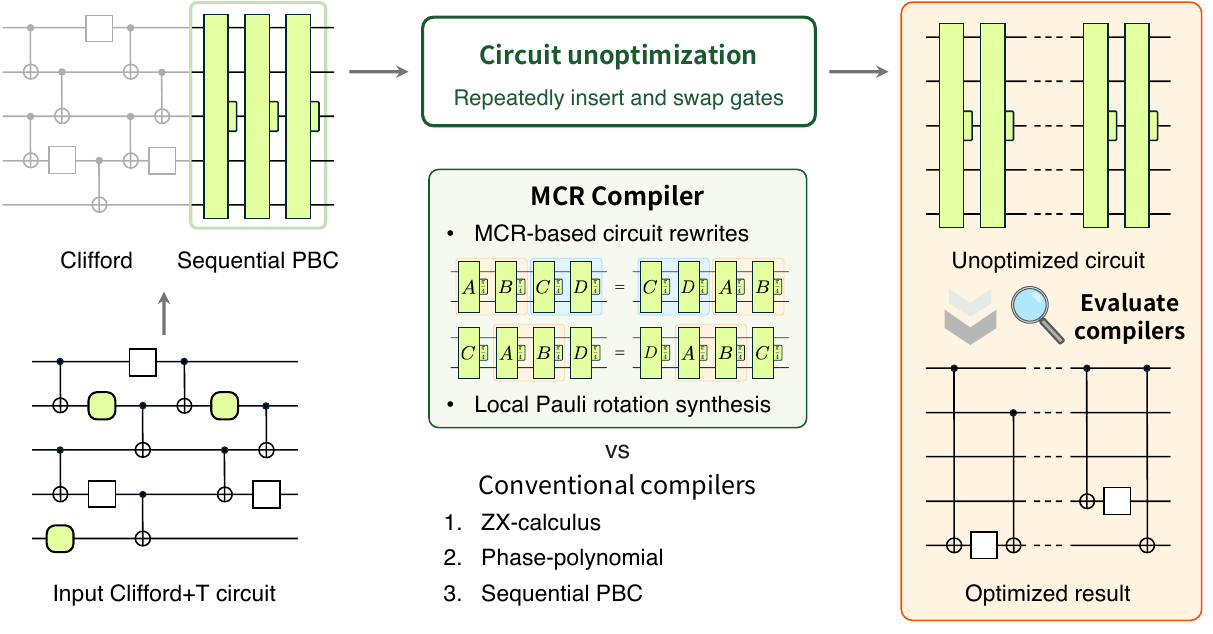}
  \caption{An overview of the quantum compiler benchmarking procedure for Clifford$+T$ circuits using MCR-based unoptimization. Starting from an input Clifford$+T$ circuit, we first reorder all $T$~gates with adjacent Clifford gates to separate the Clifford part and yield a sequential Pauli-based computation. Then, we apply the multiproduct commutation relation to the non-Clifford block, intentionally introducing redundant gates and yielding an \textit{unoptimized} circuit that is functionally equivalent but has a larger $T$-count.
    A target quantum compiler then attempts to compress this circuit and produces an optimized output, whose $T$-count is compared with that of the original circuit. If the compiler fails to recover to the original optimal $T$-count, analyzing the gap makes it possible to improve the compiler's performance.
    Using this framework, we compare the $T$-count optimization performance of the proposed multiproduct commutation Compiler with that of conventional compilers, including those based on the ZX-calculus, phase-polynomial, and sequential Pauli-based computation representations.
  }
  \label{fig: overview}
\end{figure*}

Quantum computing has the potential to outperform classical computation in specific tasks such as factorization~\cite{shor}, unstructured database search~\cite{grover}, and simulation of quantum many-body systems~\cite{SLloyd, quantum_simulation}.
While these algorithms promise significant theoretical speedups, their practical implementation requires expressing them as quantum circuits that meet the constraints of physical quantum hardware.
The ongoing increase in the number of physical qubits and the improvement of gate fidelity bring the construction of fault-tolerant quantum computers (FTQC) closer to practicality~\cite{ftqc, google_qec}.
In this emerging landscape, addressing resource constraints becomes increasingly essential, particularly the high implementation cost of specific gates in FTQC.

In FTQC architectures, the Clifford$+T$ gate set stands out as a leading candidate for universal quantum computation~\cite{nielsen_chuang} due to its compatibility with quantum error correction. Clifford gates are commonly implemented using transversal operations, which are relatively low-cost.
However, non-Clifford gates such as the $T$~gate are substantially more expensive to implement fault-tolerantly, primarily due to the need for magic state distillation~\cite{distillation}.
Therefore, reducing the number of $T$~gates in quantum circuits, a process known as \textit{$T$-count optimization}, has emerged as a key challenge for the efficient implementation of quantum algorithms on FTQC.

To enable efficient $T$-count reduction, quantum compilers form an essential component of the quantum computing stack, serving as an intermediate layer between high-level descriptions of quantum algorithms and low-level, hardware-specific implementations~\cite{quantum_compiler, compiler_toolchain, chong_qc_programming}. This design is derived from the role of classical compilers, which translate high-level programming languages into machine-level instructions~\cite{classical_compilers}.
Quantum compilers perform not only gate count reduction in quantum circuits but also a range of tasks, including gate decomposition~\cite{shende, solovay_kitaev, ross_selinger} and qubit mapping for hardware compatibility~\cite{qubit_routing, qiskit, pytket, staq}.

From the perspective of circuit optimization, a common strategy is to rewrite a quantum circuit into an intermediate representation (IR), apply optimization rules specialized for that representation, and translate it back into a quantum circuit.
Several IRs have been proposed for $T$-count reduction.
The first is the \textit{ZX-calculus}, a graphical framework which expresses quantum circuits as diagrams and simplifies them using rewrite rules such as spider fusion~\cite{zx_original_new, graph_theoretic_zx_culculus} and phase gadgetization~\cite{phase_gadget_zx_calculus}. The latter technique enables specific nonlocal optimizations when the involved gadgets can be propagated through the ZX-diagram. Nonetheless, the effectiveness of these techniques remains limited due to the presence of noncommuting gates within the circuit, such as internal Hadamard gates, which restrict simplification to local transformations.
Another well-known approach is the \textit{phase-polynomial}, an algebraic representation which encodes $CNOT+T$ circuits as linear Boolean functions and enables algorithmic optimization methods~\cite{amy_meet_in_the_middle, amy_matroid, amy_reed_muller}.
Extensions of this technique to Clifford$+T$ circuits have been proposed, a process known as \textit{Hadamard gadgetization}~\cite{campbell_compiler, fasttodd, alphatensor}.
Though this approach enables representing the entire circuit within the $CNOT+T$ framework, it inevitably introduces additional ancilla qubits, typically requiring one ancilla qubit per Hadamard gate. In the context of FTQC, these ancilla qubits must be implemented as fault-tolerant logical qubits, further increasing the resource requirements.

As an approach that can potentially overcome the limitations of the methods mentioned above, the \textit{sequential Pauli-based Computation (sequential PBC)} expresses a Clifford+$T$ circuit as a sequence of $\pm \pi/4$ multi-Pauli rotation gates obtained by rearranging all Clifford gates to the beginning of the circuit while keeping the overall equivalence~\cite{gosset}.
Notably, originally proposed as a framework for performing lattice surgery in surface codes, sequential PBC has been shown to be not only effective for reducing the $T$-count but also well-suited for optimizations targeting the instruction sets of FTQC~\cite{litinski, hirano_pbc}.
In the context of $T$-count reduction, sequential PBC has demonstrated performance comparable to other IRs, as exemplified by the compiler, TMerge~\cite{zhang_algorithm, fast_tmerge}. However, its optimization rules are currently restricted to mutually commuting Pauli rotations and thus cannot be directly applied to noncommuting cases.

In this paper, to extend the domain of existing circuit rewrite rules, we introduce a class of nonlocal and ancilla-free circuit transformations for quantum circuits, the \textit{multiproduct commutation relation (MCR)}, which is defined as a sequential PBC format.
MCR constructs a sequence consisting of multi-Pauli rotation gates that anticommute with each other but commute collectively.
As a result, MCR enables gate reorderings that go beyond local rewrites based solely on commutation relations between individual gates.
Leveraging this transformation rule, we propose the MCR Compiler that explicitly incorporates MCR as an optimization pass for $T$-count reduction. To evaluate its effectiveness in comparison with existing compilers, we construct benchmark circuits that are deliberately made redundant by MCR.
To this end, we employ MCR in an algorithm called \textit{quantum circuit unoptimization}~\cite{unoptimization}, which intentionally introduces redundancy to the circuit while keeping its equivalence.
Using this process, we can guarantee the optimal $T$-count from the original quantum circuit, allowing us to quantitatively analyze the performance of compilers in optimizing redundant $T$~gates.
Using the unoptimized circuits as benchmarks, where redundant $T$~gates are introduced by MCR, we evaluate the performance of the MCR Compiler and four types of $T$~gate optimization compilers, specifically those based on the ZX-calculus, phase-polynomial, and sequential PBC representations.
The schematic workflow of the quantum compiler benchmarking framework is shown in Fig.~\ref{fig: overview}.
Our results show that these state-of-the-art compilers, despite their different optimization frameworks, fail to reduce the $T$-count comparable to that of the original circuit.
In contrast to existing compilers, the MCR Compiler successfully exploits MCR-based transformations and achieves substantial $T$-count reductions.
These results confirm that MCR is not only a nontrivial equivalent transformation rule, but also functions as a practically effective optimization rule that enables more globally aware $T$-count reduction.

The structure of this paper is as follows.
Sec.~\ref{sec: preliminary} provides the preliminary, introducing the Pauli and Clifford groups as well as a set of rotation axes used to define multi-Pauli rotation gates.
Also, we describe a method for separating the Clifford$+T$ circuit into a sequence only composed of Clifford gates and a sequential PBC derived from $T$~gates.
In Sec.~\ref{sec: theory}, we define and formulate a nontrivial gate commutation rule, MCR, which enables transformations not evident from adjacent gate commutativity alone. We also provide a condition for constructing MCR and introduce the MCR-aware optimization pass, including its algorithmic procedure.
Sec.~\ref{sec: benchmark} verifies the nonobviousness of MCR through compiler benchmarks based on quantum circuit unoptimization, and demonstrates the advantage of MCR-based optimization over existing compilers using the MCR Compiler developed in this work.
Finally, Sec.~\ref{sec: conclusion} concludes our study and outlines potential directions for future research.

\section{Preliminary}
\label{sec: preliminary}

\subsection{Multi-Pauli rotations}
\label{subsec: rotation_set}

We begin with defining the $n$-qubit \textit{Pauli group $\mP_n$} as the group generated by $n$-fold tensor products of the single-qubit Pauli operators $\qty{I, X, Y, Z}$, multiplied by coefficients from the set $\qty{\pm1, \pm i}$,
\begin{equation}
  \mP_n = \qty{\pm 1, \pm i} \times \qty{I, X, Y, Z}^{\otimes n}.
  \label{eq:Pauli_group}
\end{equation}
We also use the \textit{Clifford group $\mC_n$} as the set of unitary operators that map every element of $\mP_n$ to another element of $\mP_n$ under conjugation,
\begin{equation}
  \mC_n = \qty{C \mid CPC^{\dagger} \in \mP_n, \, \forall P \in \mP_n}.
  \label{eq:clifford_condition}
\end{equation}
Representative elements of the Clifford group include the Hadamard ($H$), phase ($S$), and CNOT gates,
\begin{align}
  \label{eq:Hadamard}
  H               & = \frac{1}{\sqrt{2}}
  \begin{pmatrix}
    1 & 1  \\
    1 & -1
  \end{pmatrix},                         \\
  \label{eq:S}
  S               & =
  \begin{pmatrix}
    1 & 0 \\
    0 & i
  \end{pmatrix},                         \\
  \label{eq:CNOT}
  \quad{\rm CNOT} & =
  \begin{pmatrix}
    1 & 0 & 0 & 0 \\
    0 & 1 & 0 & 0 \\
    0 & 0 & 0 & 1 \\
    0 & 0 & 1 & 0
  \end{pmatrix}.
\end{align}
Clifford gates alone are not \textit{universal}, meaning that they cannot implement arbitrary unitary operations on quantum circuits~\cite{nielsen_chuang}. To achieve universality, it is necessary to include non-Clifford gates, and a widely used example is the $T$~gate,
\begin{equation}
  \label{eq:T}
  T =
  \begin{pmatrix}
    1 & 0                   \\
    0 & e^{i \frac{\pi}{4}}
  \end{pmatrix}.
\end{equation}

Both the Clifford and $T$~gates can be expressed as products of \textit{multi-Pauli rotation gates}. They are unitary operations of the form
\begin{equation}
  R_{P}\qty(\theta) = e^{-i\theta P /2} = \cos \qty(\frac{\theta}{2})I^{\otimes n} -i\sin \qty(\frac{\theta}{2})P,
  \label{eq:general_rot}
\end{equation}
where $\theta \in (-\pi, \pi]$ is the \textit{rotation angle}, and $P$ is the \textit{rotation axis}, chosen from a set of rotation axes $\mP_n^* \subset \mP_n$, defined as
\begin{equation}
  \mP_n^* \coloneqq \qty{ \pm 1 } \times \qty{I, X, Y, Z}^{\otimes n} \setminus \qty{\pm I^{\otimes n}}.
  \label{eq:rotation_axis_set}
\end{equation}
This set consists of Hermitian Pauli operators excluding the identity $\pm I^{\otimes n}$.
All Clifford operators can be rewritten as sequences of multi-Pauli rotations with angles of the form $k \pi /2$ for some integer $k$.
For example, the Hadamard gate admits the decomposition $\exp(i \pi/2)R_{Z}\qty(\pi/2)R_{X}\qty(\pi/2)R_{Z}\qty(\pi/2)$, and similar expressions exist for $S$ and CNOT. On the other hand, $T$~gate is expressed as a $\pi/4$ multi-Pauli rotation gate $\exp(i \pi/8) R_{Z}\qty(\pi/4)$.
We sometimes denote the multi-qubit Pauli rotation gates with identity operators ($I$) appended as needed to ensure consistency with the total number of qubits in the circuit, e.g., $R_{IXII}(\pi/4) = I R_{X}(\pi/4) II$.

In this paper, we consider a Clifford$+T$ circuit that contains many multi-Pauli rotation gates. To clearly illustrate their structure, we adopt a diagrammatic representation that explicitly shows the rotation axis and angle for each gate~\cite{litinski}. For each qubit, the rotation axis $\qty(\mathbbm{1}, X, Y, Z)$ is written inside a box, and the rotation angle is placed outside. Here, $\mathbbm{1}$ denotes the identity operator $I$. Gates with a rotation angle of $\pm \pi/4$ (i.e., non-Clifford gates) are colored green, while those with a phase of $k \pi /2$ (i.e., Clifford gates) are colored orange. For example, $T$~gate, $S ^{\dagger}$ gate, and $R_{XIZ}\qty(\pi/4)$ are visualized according to this rule as shown in Fig.~\ref{fig:example_litinski}.

Each multi-Pauli rotation gate $R_P\qty(\theta)$ can be decomposed into a gate sequence consisting of a single-qubit $Z$-rotation gate and Clifford gates from the set $\qty{ CNOT, H, S, S^{\dagger} }$.
This decomposition can be achieved in two steps.
First, any rotation of the form $R_{Z \dots Z}\qty(\theta)$ is reduced to a single-qubit $Z$-rotation acting on one of the qubits, combined with CNOT gates. Figure~\ref{fig:zzzzrot} illustrates an example of the transformation of the four-qubit gate $R_{ZZZZ}\qty(\theta)$.
Next, to handle a general Pauli axis $P \in \mP_n^*$, we apply local Clifford gates before and after the $Z$-rotation to convert it into a rotation around $P$.
For example, to change the rotation axis from $Z$ to $X$, we apply Hadamard gates before and after. To convert it into $Y$, we use $\left\{ S^{\dagger}, H \right\}$ before and $\left\{ H, S \right\}$ after.
Using these steps, a gate such as $R_{YXZ}\qty(\theta)$ can be equivalently transformed as shown in Fig.~\ref{fig:yxzrot}.

\begin{figure}[t]
  \centering
  \includegraphics[width=.75\linewidth]{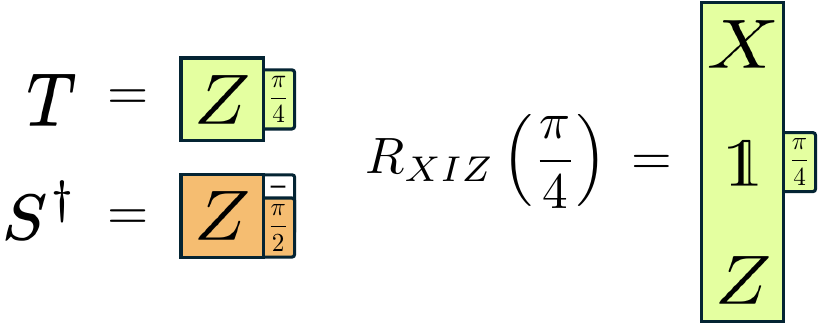}
  \caption{Representation of quantum gates used in this paper. The rotation axis is indicated inside each box, and the rotation angle is shown outside. Gates with rotation angles of $\pm\pi/4$ (non-Clifford) are colored green, while those with $k \pi /2 \ (k \in \mathbb{Z}$, Clifford) are colored orange.}
  \label{fig:example_litinski}
\end{figure}

\begin{figure}[t]
  \centering
  \includegraphics[width=0.70\linewidth]{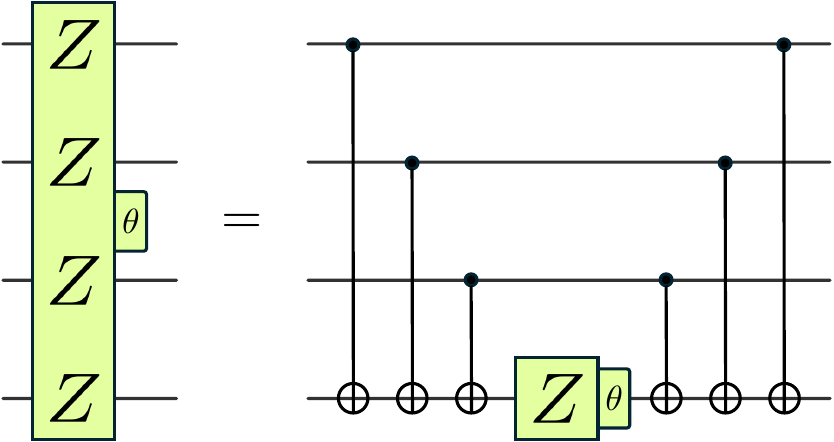}
  \caption{Decomposition of the multi-Pauli rotation gate $R_{ZZZZ}\qty(\theta)$.
    The gate applies a rotation about the tensor product of four Pauli-$Z$ operators, and is decomposed into a standard circuit using a sequence of CNOT gates and a single-qubit $Z$-rotation.}
  \label{fig:zzzzrot}
\end{figure}
\begin{figure}[t]
  \centering
  \includegraphics[width=0.93\linewidth]{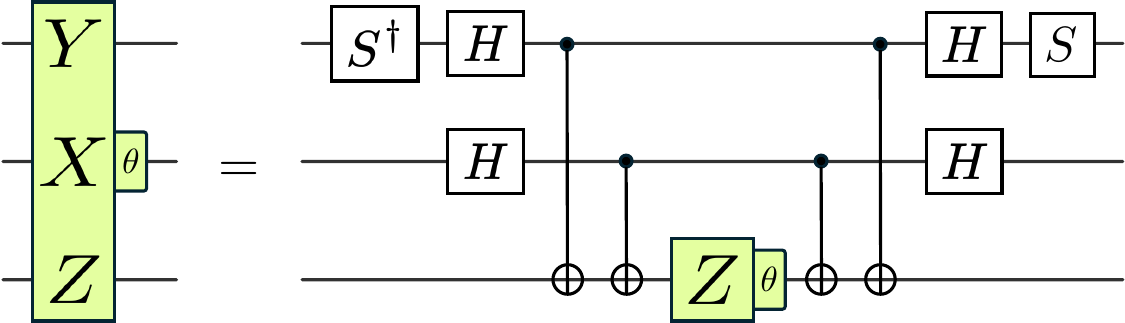}
  \caption{Decomposition of the three-qubit rotation gate $R_{YXZ}\qty(\theta)$. The gate implements a rotation about the tensor product of $Y$, $X$, and $Z$, and is decomposed into a standard circuit using local Clifford gates and a single $Z$-rotation.}
  \label{fig:yxzrot}
\end{figure}

\subsection{Sequential PBC}
\label{subsec: sep_circuits}

We separate a given Clifford$+T$ circuit into Clifford gates and non-Clifford multi-Pauli rotations as introduced in Sec.~\ref{subsec: rotation_set}.
The order of operation between a Clifford gate $C$ and a $T$~gate can generally be exchanged by replacing the $T$~gate with a modified gate $\widetilde{T}$ defined as
\begin{equation}
  \begin{split}
    \widetilde{T} & = CTC^{\dagger}                                              \\
                  & = e^{i \frac{\pi}{8}} C R_Z\qty(\frac{\pi}{4}) C^{\dagger}   \\
                  & = e^{i \frac{\pi}{8}} R_{CZC^{\dagger}} \qty(\frac{\pi}{4}).
    \label{eq:axis_update}
  \end{split}
\end{equation}
By repeatedly applying this operation to adjacent Clifford gates, any $T$~gate can be propagated to the right end of the circuit.
Figure~\ref{fig:t_clifford_commutation} illustrates an example of $T$ propagation.
Using Eq.~\eqref{eq:axis_update}, this operation yields a new rotation axis $-X Y X$ in the end.

\begin{figure}[t]
  \includegraphics[width=1.0\linewidth]{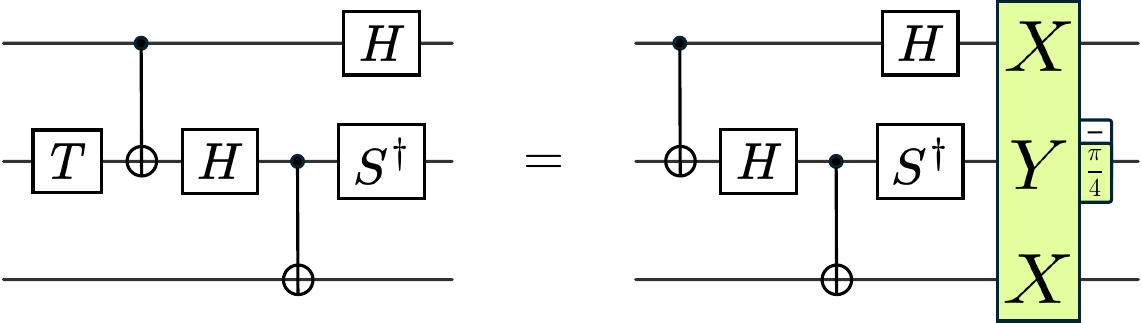}
  \caption{Transformation of a circuit segment containing one $T$~gate surrounded by Clifford gates. In the transformed circuit, the Clifford gates remain unchanged, and a multi-Pauli rotation gate $R_{XYX}\qty(-\pi/4)$ is appended at the right end.
  }
  \label{fig:t_clifford_commutation}
\end{figure}

By repeating this procedure for all $T$~gates in a Clifford$+T$ circuit, we can express the circuit as the product of Clifford gates followed by the \textit{sequential Pauli-based Computation (sequential PBC)} without loss of generality~\cite{gosset}.
That is, given a Clifford$+T$ circuit $U = e^{i\phi}
  C_0'
  T^{\otimes l_1}
  C_1'
  T^{\otimes l_2}
  C_2'
  \cdots
  C_{k-1}'
  T^{\otimes l_k}
  C_{k}'$, where $\phi$ is a global phase, and $T^{\otimes l_i} (i = 1, \cdots , k)$ denotes the $T$~gates acting on $l_i$ qubits between the $(i-1)$th and $i$th Clifford sequences, we can convert it into the following form
\begin{equation}
  U = e^{i\phi} \qty{\prod_{i=1}^k R_{P_i} \qty(\frac{\pi}{4})} C_k.
  \label{eq:litinski_format}
\end{equation}
In this paper, we focus primarily on the non-Clifford part, as it determines the $T$-count of the circuit.
Note that the above procedure can also be applied to Clifford$+R_Z$ circuits, i.e., circuits that consist of Clifford gates and $Z$-rotation gates with arbitrary angles, without loss of generality.

\section{Theory of multiproduct commutation relation}
\label{sec: theory}

In this section, we present the multiproduct commutation relation (MCR), a nontrivial commutation rule that enables certain groups of noncommuting multi-Pauli rotation gates to be reordered. We also establish the necessary and sufficient conditions for MCR construction and show how MCR enables circuit simplification beyond the reach of local optimization rules.
Then, we propose a compiler explicitly including optimization passes derived from MCR.

\subsection{Definition of MCR}
\label{subsec: definition_mcr}

\begin{figure}[t]
  \centering
  \includegraphics[width=0.95\linewidth]{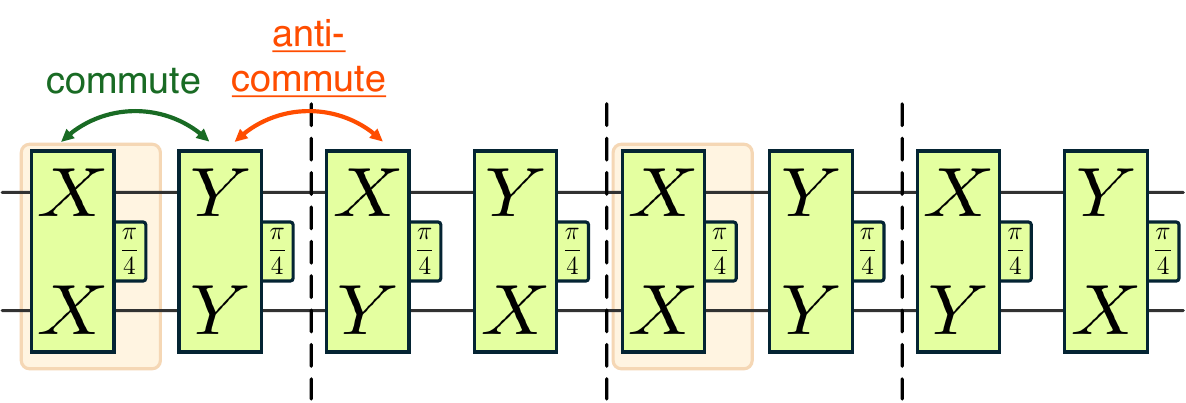}
  \caption{A circuit containing adjacent Pauli rotation gates, some of which commute. However, due to intervening anticommuting gates, no further gate merging for circuit optimization is possible.}
  \label{fig:nontrivial_example}
\end{figure}

We define each segment of sequential PBC, separated by dotted lines in Fig.~\ref{fig:nontrivial_example}, as a \textit{$T$ layer}~\cite{litinski}.
\begin{definition}[$T$ layer]
  \label{def:T_layer}
  In sequential PBC, a segment is called a $T$ layer if
  \begin{enumerate}[label=(\roman*)]
    \item All multi-Pauli rotation gates within the layer mutually commute.
    \item For any multi-Pauli rotation around axis $P$ in the $(i+1)$th layer, there exists at least one multi-Pauli rotation around axis $Q$ in the $i$th layer such that $Q$ anticommutes with $P$.
  \end{enumerate}
\end{definition}

If the rotation axes of non-Clifford gates match exactly in succession, their rotation angles can be combined into a single gate, resulting in $T$-count reduction by two, i.e., $R_{P} \qty(\pm \pi/4) R_{P} \qty(\pm \pi/4) = R_{P} \qty(k\pi/2) \in \mC_n$.
Such merging optimization applies even to nonadjacent gates, provided they belong to the same $T$ layer.
$T$-count optimization compilers exploiting this gate synthesis rule in common $T$ layers have been proposed~\cite{litinski, zhang_algorithm}.

However, some circuits remain unsimplifiable under this approach.
For example, the circuit in Fig.~\ref{fig:nontrivial_example} contains four $T$ layers.
Within each layer, gates such as $R_{XX}\qty(\pi/4)$ and $R_{YY}\qty(\pi/4)$ or $R_{XY}\qty(\pi/4)$ and $R_{YX}\qty(\pi/4)$ commute with each other. In contrast, gates across adjacent layers such as $R_{YY}\qty(\pi/4)$ and $R_{XY}\qty(\pi/4)$ anticommute, making them unswappable and preventing further optimization.
Applying MCR overcomes this limitation, i.e., certain gates can be swapped as composite units while preserving the overall circuit functionality.
In the remainder of this section, we formalize this commutation rule in terms of multi-Pauli rotation axes and demonstrate its use in enabling nontrivial gate reordering.

We define the \textit{multiproduct commutation relation} (MCR) as a specific combination of multi-Pauli rotation gates whose axes satisfy certain commutation and anticommutation conditions.
\begin{definition}[Multi-product commutation relation]
  \label{def:multi_commutation}
  \par
  Let $A, B, C, D \in \mP_n^*$ be mutually distinct (up to a sign $\pm 1$) multi-Pauli axes corresponding to the rotation gates with an angle of $\pi/4$. We say that the pairs $\qty(R_A, R_B)$ and $\qty(R_C, R_D)$ satisfy the multiproduct commutation relation (MCR) if the following three conditions hold:
  \begin{enumerate}
    \item $\commutator{A}{B} = \commutator{C}{D}=0$,
    \item $\anticommutator{A}{C} = \anticommutator{A}{D}=\anticommutator{B}{C} = \anticommutator{B}{D}=0$,
    \item $\commutator{A+B}{C+D} =0$.
  \end{enumerate}
  Here, $\commutator{M_1}{M_2} \coloneqq M_1 M_2 - M_2 M_1$ denotes the commutator and $\anticommutator{M_1}{M_2} \coloneqq M_1 M_2 + M_2 M_1$ denotes the anticommutator.
\end{definition}
\begin{theorem}
  \label{thm:multi_product_commutable}
  Any sequential PBC whose axes satisfy the MCR can be entirely swapped.
  That is, for $\qty(R_A, R_B)$ and $\qty(R_C, R_D)$ satisfying Definition~\ref{def:multi_commutation}, Eq.~\eqref{eq:multi_product_commutation} holds as shown in
  Fig.~\ref{fig:MCR_ABCD},
  \begin{equation}
    \begin{split}
       & R_{D}\qty( \frac{\pi}{4} )
      R_{C}\qty( \frac{\pi}{4} )
      R_{B}\qty( \frac{\pi}{4} )
      R_{A}\qty( \frac{\pi}{4} )    \\
       & =
      R_{B}\qty( \frac{\pi}{4} )
      R_{A}\qty( \frac{\pi}{4} )
      R_{D}\qty( \frac{\pi}{4} )
      R_{C}\qty( \frac{\pi}{4} )
    \end{split}
    \label{eq:multi_product_commutation}
  \end{equation}
  \begin{figure}[h]
    \centering
    \includegraphics[width=.95\linewidth]{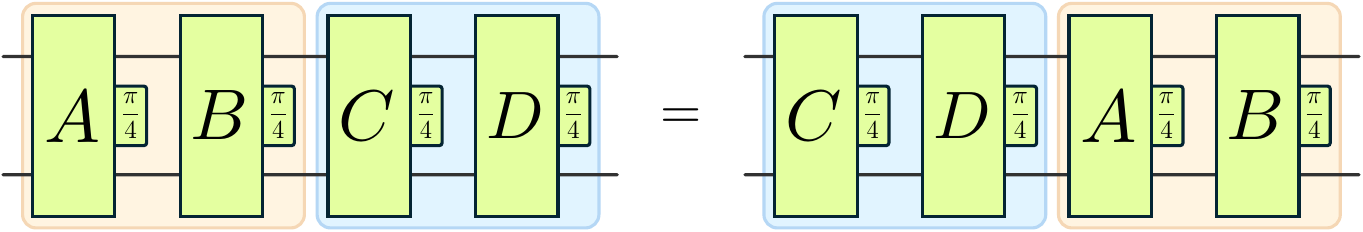}
    \caption{Circuit diagram illustrating the gate sequence satisfying multiproduct commutation relation.}
    \label{fig:MCR_ABCD}
  \end{figure}
\end{theorem}

\begin{proof}
  From conditions 1 and 3 of Definition~\ref{def:multi_commutation}, we have
  \begin{equation}
    \label{eq:multi_product_commutable_proof}
    \begin{split}
       & \quad  R_{D}\qty( \frac{\pi}{4} )
      R_{C}\qty( \frac{\pi}{4} )
      R_{B}\qty( \frac{\pi}{4} )
      R_{A}\qty( \frac{\pi}{4} )                                                                               \\
       & =e^{-i\frac{\pi}{8} \qty(D+C)} e^{-i\frac{\pi}{8} \qty(B+A)}                                          \\
       & =e^{-i\frac{\pi}{8}\qty(B+A)}e^{-i\frac{\pi}{8}\qty(D+C)}                                             \\
       & =R_{B}\qty(\frac{\pi}{4}) R_{A}\qty(\frac{\pi}{4}) R_{D}\qty(\frac{\pi}{4}) R_{C}\qty(\frac{\pi}{4}).
    \end{split}
  \end{equation}
\end{proof}
Note that condition~2 of Definition~\ref{def:multi_commutation} is not explicitly used in the above derivation. Instead, it serves to exclude the trivial case where all pairs among $R_A, R_B, R_C$, and $R_D$ commute.
In this case, the gates can be freely reordered, and the commutation relation becomes obvious. Condition~2 ensures that the MCR captures only nontrivial gate reorderings that cannot be justified by pairwise commutation alone, which follows from the commutation and anticommutation relations of the underlying Pauli operators.

Consider the circuit shown in Fig.~\ref{fig:nontrivial_circuit}, where Theorem~\ref{thm:multi_product_commutable} enables a nontrivial gate reordering that leads to further optimization.
The circuit remains equivalent even if the pairs $(R_{XY}, R_{YX})$ and $(R_{XX}, R_{YY})$ are swapped, since these four gates satisfy all conditions in Definition~\ref{def:multi_commutation}.
After performing this swap, the sequence of rotation axes becomes $XX, YY, XX, YY, XY, YX, XY, YX$. Since $\commutator{X X}{Y Y}=0$ and $\commutator{X Y}{Y X}=0$, the respective pairs are pairwise commutable. This observation allows each adjacent pair to be merged into a single Clifford gate, resulting in the simplified sequence: $R_{XX}(\pi/2), R_{YY}(\pi/2), R_{XY}(\pi/2), R_{YX}(\pi/2)$.
This example illustrates how MCR enables circuit simplification beyond local rewrite rules by capturing a broader class of gate rearrangement of multi-Pauli rotations.

\begin{figure}[t]
  \centering
  \includegraphics[width=.95\linewidth]{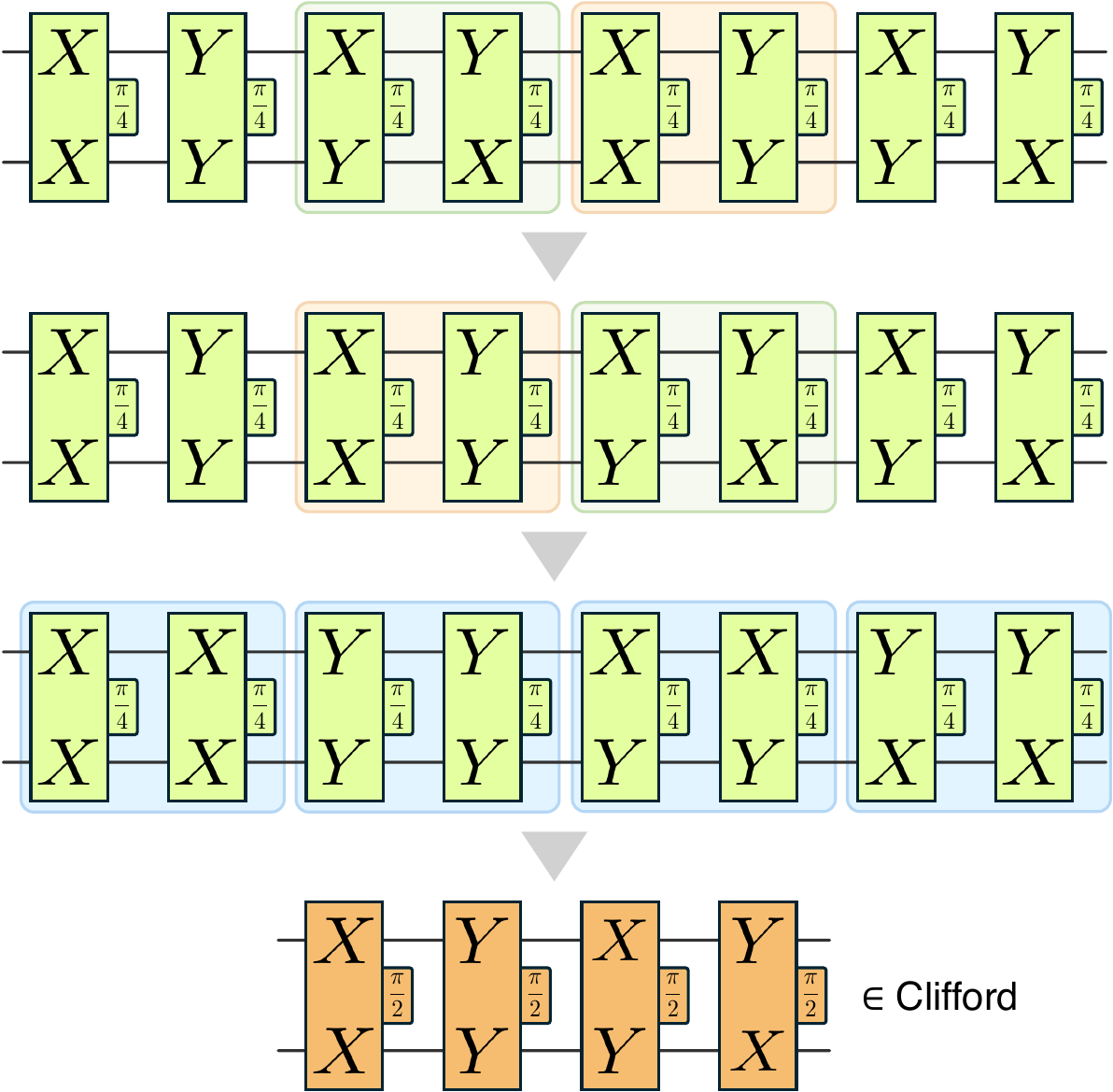}
  \caption{Optimization of the circuit shown in Fig.~\ref{fig:nontrivial_example} using multiproduct commutation relation condition. The circuit remains equivalent under the swap of commuting pairs, $(R_{XX}(\pi/4), R_{YY}(\pi/4))$ and $(R_{XY}(\pi/4), R_{YX}(\pi/4))$. After rearrangement, the circuit is partitioned into two $T$ layers consisting of four gates each. Within each $T$ layer, gates with shared rotation axes can be merged into Clifford gates. As a result, the entire circuit is reduced to a form that can be expressed using only Clifford operations.}
  \label{fig:nontrivial_circuit}
\end{figure}

A natural question is whether the gate sequence satisfying MCR arises in quantum circuits for practically relevant algorithms.
In fact, such structures appear ubiquitously in Trotterized quantum circuits for simulating quantum many-body spin dynamics, a class of algorithms that has been explored as a benchmark task demonstrating quantum advantage~\cite{quantum_utility,yoshioka2024hunting, arute_fermi_hubbard,alam_2d_72_qubits,yu_spin_chains}.
As a representative example, we focus on the kicked spin dynamics of the XY model. This system is described by the Hamiltonian
\begin{equation}
  H= J\sum_{\langle i, j \rangle} \qty(X_i X_j+ Y_i Y_j) + h_z\sum_i Z_i.
  \label{eq: xy_model}
\end{equation}
Here, $\langle i, j \rangle$ denotes nearest-neighbor pairs, $J$ is the strength of interaction, and $h_z$ represents the strength of the longitudinal magnetic field~\cite{xy_model}.
In digital quantum simulation, the time evolution governed by this Hamiltonian is implemented via Trotter decomposition, resulting in a quantum circuit with alternating layers of two-qubit rotations, consisting of $R_{XX}$ and $R_{YY}$, and layers of single-qubit $Z$-rotations $R_{Z}$. In this setting, the pairs $(R_{XX}, R_{YY})$ and $(R_{ZI}, R_{IZ})$ satisfy the conditions in Definition~\ref{def:multi_commutation}. Consequently, these gates can be rearranged such that all $Z$ terms are grouped together and subsequently merged. Based on this model, we benchmark the performance of our proposed compiler incorporating MCR-based circuit transformations (described in Sec.~\ref{subsec: mcr_aware_compiler}) and compare it with conventional optimization frameworks in Sec.~\ref{subsec: mcr_application}.
The rearrangement of $Z$-rotations follows directly from particle-number conservation in this model. More generally, even for complex many-body Pauli rotations, appropriate Clifford transformations can induce local structures that admit similar reorderings, which are systematically captured by MCR.

\subsection{Systematic design of sequential PBC satisfying MCR}
\label{subsec: mcr_construction_condition}

We now explain a method that enables efficient construction of sequential PBC satisfying MCR.
The following theorem provides the requirements under which four distinct rotation axes give rise to an MCR construction.
\begin{theorem}[Construction of rotation axes satisfying MCR]
  \label{thm:abc_multi_commutation}
  Let $A, B, C, D \in \mP_n^*$, and let all rotation angles be $\pi/4$. Here, $A, B, C, D$ are assumed to have mutually distinct rotation axes, up to a sign $\pm 1$.
  Suppose $A, B, C$ satisfy $\commutator{A}{B}=0$ and $\anticommutator{A}{C}=\anticommutator{B}{C}=0$.
  Then, the following (I) and (II) are equivalent.
  \begin{enumerate}[label=(\Roman*)]
    \item $A, B, C$, and $D$ satisfy all of the following (i)--(iii):
          \begin{enumerate}[label=(\roman*)]
            \item $\commutator{C}{D}=0$,
            \item $\anticommutator{A}{D}=\anticommutator{B}{D}=0$,
            \item $\commutator{A+B}{C+D}=0$.
          \end{enumerate}
    \item $D=-ABC$.
  \end{enumerate}
\end{theorem}

\begin{proof}
  See Appendix~\ref{sec: proof_of_mcr}.
\end{proof}
Theorem~\ref{thm:abc_multi_commutation} provides a constructive method for generating rotation axes that satisfy MCR.
For instance, in the circuit from Fig.~\ref{fig:nontrivial_example}, we can set $A=XY$, $B=YX$, and $C=XX$. Then, $-ABC =YY$,
confirming that the rotation axis of the fourth gate satisfies $D=-ABC$.
While this example involves two-qubit multi-Pauli rotation gates, we can apply MCR to $n$-qubit multi-Pauli rotation gates as well.
In this case, the total number of combinations of rotation axes that satisfy MCR grows exponentially with $n$. For details, see Appendix~\ref{sec: mcr_scaling}.

\subsection{MCR Compiler}
\label{subsec: mcr_aware_compiler}

We develop the MCR Compiler, which explicitly incorporates two fundamental MCR-based circuit rewrites, \texttt{StandardMCR} and \texttt{EndpointMCR}, into the optimization pipeline.
The pseudocode of these passes is shown in Algorithm~\ref{alg: mcr_swap}.
Before applying these rewrites, the compiler partitions a sequential PBC into $T$~layers, as defined in Definition~\ref{def:T_layer}. This preprocessing step is denoted as \texttt{TLayerPartition} in Algorithm~\ref{alg: mcr_compiler}.

\texttt{StandardMCR} operates on two consecutive $T$~layers $L_{\rm left}$ and $L_{\rm right}$. It selects a pair of distinct rotations $R_A(\pi/4), R_B(\pi/4) \in L_{\rm left}$, and scans $L_{\rm right}$ for a rotation $R_C(\pi/4)$ whose axis anticommutes with both $A$ and $B$.
For each candidate $R_C(\pi/4)$, the algorithm checks the presence of a rotation $R_D(\pi/4) \in L_{\rm right}$ with axis $D=\pm ABC$.
If a rotation with $D=-ABC$ is found, the pairs $(R_A(\pi/4), R_B(\pi/4))$ and $(R_C(\pi/4), R_D(\pi/4))$ are swapped as depicted in Fig.~\ref{fig:MCR_ABCD}.
If instead $D=ABC$, the rotation $R_A(\pi/4)$ is decomposed into $R_A(\pi/2)$ followed by $R_A(-\pi/4)$, after which the MCR-based swapping is applied.

\texttt{EndpointMCR} operates on three consecutive $T$~layers $L_{\rm left}$, $L_{\rm middle}$, and $L_{\rm right}$, where $L_{\rm middle}$ contains exactly two rotations $R_A(\pi/4)$ and $R_B(\pi/4)$. This pass searches for a rotation $R_C(\pi/4) \in L_{\rm left}$ whose rotation axis anticommutes with both $A$ and $B$. Then, it checks whether there exists a $R_D(\pi/4) \in L_{\rm right}$ with $D=\pm ABC$.
If a rotation with $D=ABC$ is found, the position of $R_C(\pi/4)$ and $R_D(\pi/4)$ can be exchanged directly as illustrated in Fig.~\ref{fig:endpoint_mcr}.
If instead $D=-ABC$, the rotation $R_C(\pi/4)$ is decomposed into $R_C(\pi/2)$ followed by $R_C(-\pi/4)$, after which the same MCR-based swapping is applied.
Note that the sign condition for introducing an additional rotation with an angle of $\pi/2$ is reversed compared to \texttt{StandardMCR}.

\begin{figure}[htbp]
  \begin{algorithm}[H]
    \caption{MCR-based optimization passes}
    \begin{algorithmic}[1]
      \Function{StandardMCR}{$L_{\rm left}, L_{\rm right}$}
      \ForAll{\textbf{distinct} $R_A(\pi/4), R_B(\pi/4) \in L_{\rm left}$}
      \ForAll{$R_C(\pi/4)\in L_{\rm right}$}
      \If{$\anticommutator{A}{C}=\anticommutator{B}{C}=0$}
      \ForAll{$R_D(\pi/4) \in L_{\rm right}$}
      \If{$D=-ABC$}
      \State $L_{\rm left}$.remove($[R_A(\pi/4), R_B(\pi/4)]$)
      \State $L_{\rm left}$.append($[R_C(\pi/4), R_D(\pi/4)]$)
      \State $L_{\rm right}$.remove($[R_C(\pi/4), R_D(\pi/4)]$)
      \State $L_{\rm right}$.prepend($[R_A(\pi/4), R_B(\pi/4)]$)
      \State \Return $\qty(L_{\rm left}, L_{\rm right})$
      \ElsIf{$D=ABC$}
      \State $L_{\rm left}$.remove($[R_A(\pi/4), R_B(\pi/4)]$)
      \State $L_{\rm left}$.append($[R_A(\pi/2), R_C(\pi/4), R_D(\pi/4)]$)
      \State $L_{\rm right}$.remove($[R_C(\pi/4), R_D(\pi/4)]$)
      \State $L_{\rm right}$.prepend($[R_A(-\pi/4), R_B(\pi/4)]$)
      \State \Return $\qty(L_{\rm left}, L_{\rm right})$
      \EndIf
      \EndFor
      \EndIf
      \EndFor
      \EndFor
      \State \Return $\qty(L_{\rm left}, L_{\rm right})$
      \EndFunction
      \State
      \Function{EndpointMCR}{$L_{\rm left}, L_{\rm middle}, L_{\rm right}$}
      \If{$\lvert L_{\rm middle} \rvert \neq 2$}

      \State \Return $\qty(L_{\rm left}, L_{\rm middle}, L_{\rm right})$
      \EndIf
      \State $\qty(R_A(\pi/4), R_B(\pi/4)) \gets$ the two rotations in $L_{\rm middle}$
      \ForAll{$R_C(\pi/4) \in L_{\rm left}$}
      \If{$\anticommutator{A}{C}=\anticommutator{B}{C}=0$}
      \ForAll{$R_D(\pi/4) \in L_{\rm right}$}
      \If{$D=ABC$}
      \State $L_{\rm left}$.remove($R_C(\pi/4)$)
      \State $L_{\rm left}$.append($R_D(\pi/4)$)
      \State $L_{\rm right}$.remove($R_D(\pi/4)$)
      \State $L_{\rm right}$.prepend($R_C(\pi/4)$)
      \State \Return $\qty(L_{\rm left}, L_{\rm middle}, L_{\rm right})$
      \ElsIf{$D=-ABC$}
      \State $L_{\rm left}$.remove($R_C(\pi/4)$)
      \State $L_{\rm left}$.append($[R_C(\pi/2), R_D(\pi/4)]$)
      \State $L_{\rm right}$.remove($R_D(\pi/4)$)
      \State $L_{\rm right}$.prepend($R_C(-\pi/4)$)
      \State \Return $\qty(L_{\rm left}, L_{\rm middle}, L_{\rm right})$
      \EndIf
      \EndFor
      \EndIf
      \EndFor
      \State \Return $\qty(L_{\rm left}, L_{\rm middle}, L_{\rm right})$
      \EndFunction
    \end{algorithmic}
    \label{alg: mcr_swap}
  \end{algorithm}
\end{figure}
\begin{figure}[t]
  \begin{algorithm}[H]
    \caption{MCR Compiler}
    \begin{algorithmic}[1]
      \Input\textbf{:}
      \Desc{$S_0$}{$n$-qubit sequential PBC}
      \Desc{$N_{\rm iter}$}{Optimization iteration count}
      \EndInput
      \State $C_{\rm out} \gets$ empty circuit
      \State $S \gets S_0$
      \For{$t = 0$ to $N_{\rm iter}-1$}
      \State     $\qty[L_0, L_1, ..., L_{m-1}] \gets$ \texttt{TLayerPartition}($S$)
      \State $x \gets $ choose from \{\texttt{StandardMCR}, \texttt{EndpointMCR}\}
      \If{$x=\texttt{StandardMCR}$}
      \For{\textbf{each} consecutive pairs $\qty(L_j, L_{j+1})$}
      \State $\qty(L_j, L_{j+1}) \gets$\texttt{StandardMCR}($L_j, L_{j+1}$)
      \EndFor
      \Else
      \For{\textbf{each} consecutive triples $\qty(L_j, L_{j+1}, L_{j+2})$}
      \State $\qty(L_j, L_{j+1}, L_{j+2}) \gets$\texttt{EndpointMCR}($L_j, L_{j+1}, L_{j+2}$)
      \EndFor
      \EndIf
      \State     $C$, $S \gets$ \texttt{MergeAndExtract}($\qty[L_0, ..., L_{m-1}]$)
      \State     $C_{\rm out}$.append($C$)
      \EndFor
      \Output\textbf{:}
      \Desc{$C_{\rm out}$}{Extracted Clifford circuit}
      \Desc{$S$}{Optimized sequential PBC}
      \EndOutput
    \end{algorithmic}
    \label{alg: mcr_compiler}
  \end{algorithm}
\end{figure}

\begin{figure}[t]
  \centering
  \includegraphics[width=.75\linewidth]{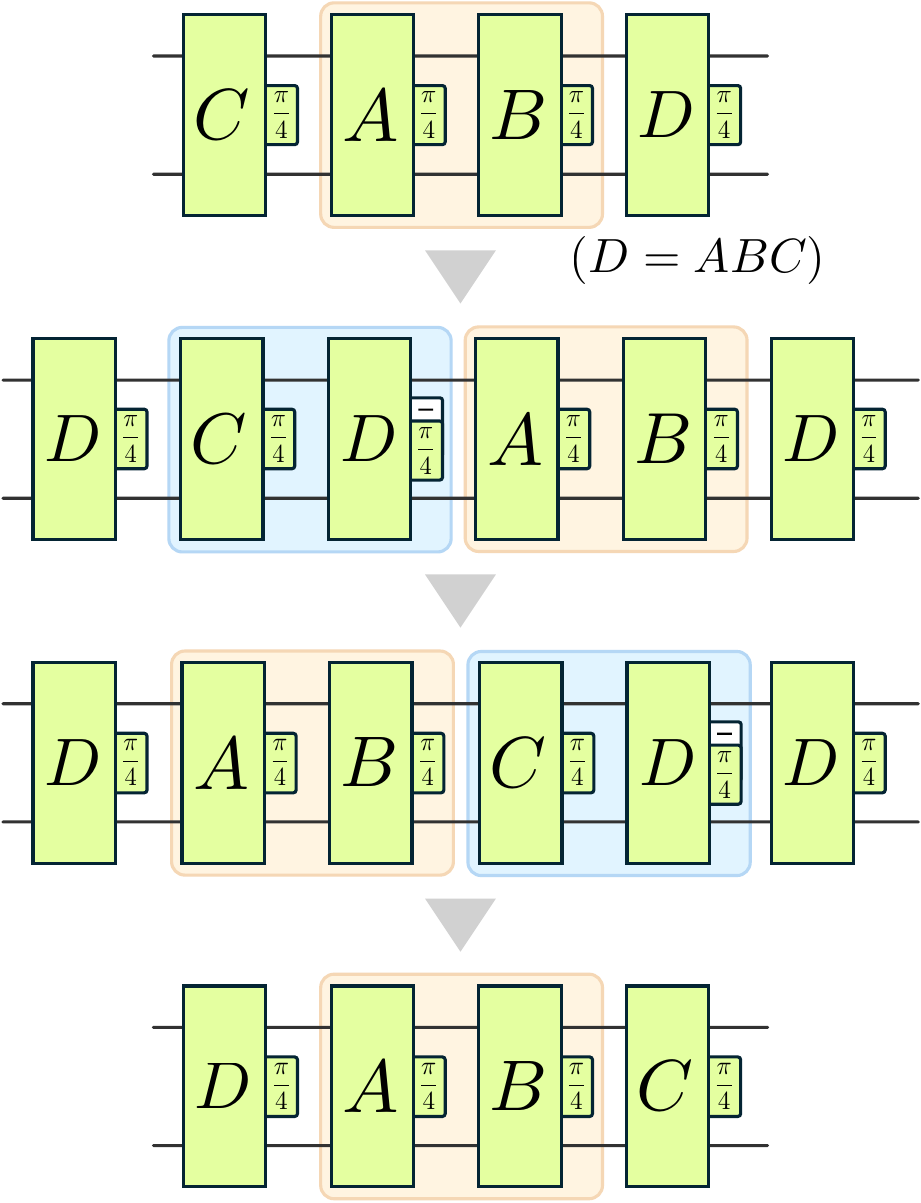}
  \caption{\texttt{EndpointMCR} for $D=ABC$.
  By introducing the sign-flipped rotation $R_D(-\pi/4)$ and then applying multiproduct commutation relation, the resulting locally synthesized circuit becomes equivalent to one in which the positions of $R_C(\pi/4)$ and $R_D(\pi/4)$ are exchanged.}
  \label{fig:endpoint_mcr}
\end{figure}

We now describe the overall workflow of the MCR Compiler, summarized in Algorithm~\ref{alg: mcr_compiler}.
Starting from an input sequential PBC $S_0$, the compiler performs a fixed number of optimization iterations $N_{\rm iter}$, each consisting of an MCR-based rewrite pass followed by local rotation synthesis.
At each iteration, the current circuit $S$ is first partitioned into $T$~layers via \texttt{TLayerPartition}.
A rewrite pass $x\in \qty{\texttt{StandardMCR}, \texttt{EndpointMCR}}$ is then selected and applied to the layers.
In our implementation, the two passes are chosen alternately, starting and ending with \texttt{EndpointMCR}.
Depending on the choice of $x$, the pass is applied either to consecutive layer pairs $\qty(L_j, L_{j+1})$ (\texttt{StandardMCR}) or to consecutive layer triples $\qty(L_j, L_{j+1}, L_{j+2})$ (\texttt{EndpointMCR}).
Rather than exploring all equivalent circuit representations globally, the MCR Compiler restricts the search to a subspace defined by the partition of the circuit into $T$~layers. Within this layered structure, \texttt{StandardMCR} and \texttt{EndpointMCR} sequentially scan candidate rewrites across neighboring layers and apply the first valid solution found, keeping the search computationally tractable.

The resulting layers are then simplified by \texttt{MergeAndExtract}, which merges successive rotations with identical axes and extracts Clifford gates as described in Sec.~\ref{subsec: sep_circuits}.
The extracted Clifford component is accumulated into $C_{\rm out}$, while the remaining non-Clifford part defines the updated sequential PBC~$S$.
In addition, the compiler may optionally insert local identity padding between optimization iterations. Specifically, when three multi-Pauli rotations $R_A(\pi/4), R_B(\pi/4)$, and $R_C(\pi/4)$ satisfy conditions~1 and 2 of Definition~\ref{def:multi_commutation}, the adjacent pair $R_{ABC}(\pi/4)$ and $R_{ABC}(-\pi/4)$ is introduced while preserving circuit equivalence. This local expansion can create additional opportunities for \texttt{StandardMCR} rewrites, thereby facilitating further simplifications by \texttt{MergeAndExtract}.

\section{Compiler benchmark using MCR}
\label{sec: benchmark}

In this section, we benchmark the MCR Compiler and investigate whether existing compilers can exploit MCR-based transformations. We first introduce the concept of quantum circuit unoptimization, an operation that transforms a circuit into a more complex one while preserving its overall functionality. We then outline the associated benchmarking framework using the $T$-count as a quantitative evaluation metric. Finally, we compare the optimization performance of the MCR Compiler with that of existing compilers.

\subsection{Quantum circuit unoptimization}
\label{subsec: def_unopt}

We use the term \textit{quantum circuit unoptimization} to describe the process of deliberately introducing gate redundancies while keeping the circuit equivalent.
The following definition formalizes this concept.

\begin{figure}[t]
  \begin{center}
    \includegraphics[width=1.0\linewidth]{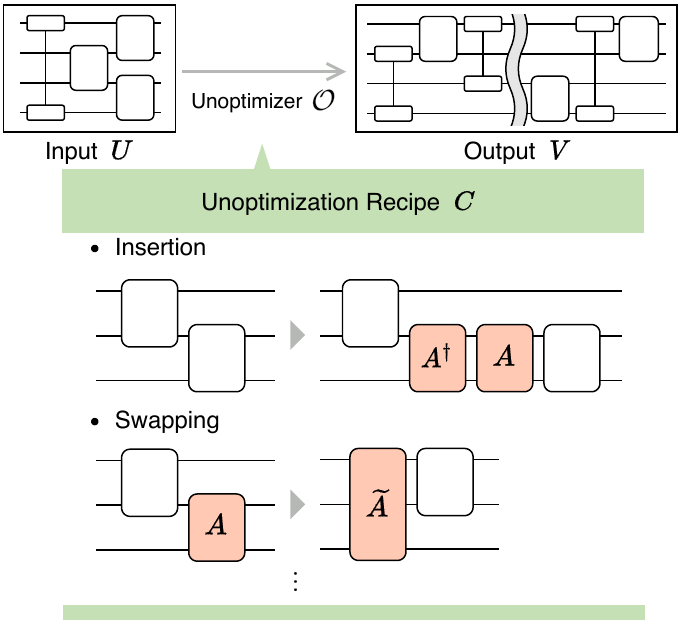}
  \end{center}
  \caption{Schematic illustration of quantum circuit unoptimization. The basic operations consist of gate insertion and gate swapping, which are used to construct an unoptimization recipe. Applying this recipe to an input circuit $U$ produces a redundant circuit $V$ that is equivalent to $U$ up to a global phase.}
  \label{fig:concept_of_unopt}
\end{figure}

\begin{definition}[Quantum circuit unoptimization~\cite{unoptimization}]
  Quantum circuit unoptimization is the process of generating a redundant $n$-qubit circuit $V$ from an $n$-qubit circuit $U$, as expressed by
  \begin{equation}
    V=\mathcal{O}\left( U, C \right).
    \label{eq:def_unopt}
  \end{equation}
  Here, $C$ is a recipe constructed from elements of the set of basic unoptimization operations $\mS$.
  $\mathcal{O}$ is a function (Unoptimizer) that transforms $U$ into $V$ based on the information in $C$.
  We assume that both $U$ and $V$ are $\poly (n)$ gates, and the operation generating $V$ from $U$ can be efficiently executed on a classical computer.
  If the conditions for unoptimization are satisfied, $V$ is equivalent to $U$ up to a global phase $\phi$, i.e.,
  \begin{align*}
    V^{\dagger}U = e^{i\phi}I,
    \label{eq:equivalence}
    \stepcounter{equation}\tag{\theequation}
  \end{align*}
  where $I$ represents the $n$-qubit identity operator.
\end{definition}
Figure~\ref{fig:concept_of_unopt} illustrates a diagram of quantum circuit unoptimization. Given an input quantum circuit $U$, a more complex quantum circuit $V$ is generated by applying a recipe that describes how to introduce redundancy.
To define such a recipe, it is necessary to determine the set of allowed basic unoptimization operations $\mS$. Note that only circuit operations that preserve the equivalence of the quantum circuit before and after the operation can be included in $\mS$.
The complexity of the unoptimized circuit depends critically on the choice of $\mS$.
For example, if we take $\mS = \qty{\text{Insert two $X$ gates at random positions in the circuit}}$, this satisfies the condition that the unoptimization operations preserve circuit equivalence. However, such redundancy is trivial: any compiler that recognizes $X^2=I$ will immediately remove the inserted gates.

To generate more meaningful redundancy, we therefore set $\mS = \qty{\text{\textit{Gate insertion, Gate swapping}}}$ and construct these operations based on MCR. Specifically, redundancies are introduced by inserting pairs of multi-Pauli rotations whose angles cancel each other along the same rotation axis, as well as by applying gate swapping as shown in Fig.~\ref{fig:MCR_ABCD}.
By combining these operations, we systematically generate circuits that are equivalent to the original circuit while having a larger $T$-count, which are used as benchmark circuits in this work.
The detailed construction procedure of these unoptimized circuits is described in Appendix~\ref{sec: mcr_unopt}.

\subsection{Numerical experiments}
\label{subsec: numerical_experiments}

We present numerical experiments evaluating how well existing compilers can optimize circuits containing MCR-based gate structures.
We focus on four $T$~gate optimization compilers: Pytket, PyZX, TMerge, and FastTODD.
To assess compiler performance, we use the unoptimized circuit $V$ generated by
Algorithm~\ref{alg: gen_unopt} in Appendix~\ref{sec: mcr_unopt}
as input.
We explain how to evaluate compiler performance using the unoptimized circuit $V$ in Algorithm~\ref{alg: benchmark}. A workflow of this benchmark task is illustrated in Fig.~\ref{fig:benchmark_example}.
Each unoptimized circuit has a $T$-count denoted as $t_{\rm unopt}$ and is passed to a compiler for optimization, resulting in a circuit with $T$-count $t_{\rm opt}$.
We quantify the optimization performance using the $T$-count reduction rate $p$ defined as
\begin{equation}
  \label{eq:reduction_rate}
  p = \frac{t_{\rm unopt}-t_{\rm opt}}{t_{\rm unopt}-t_{\rm original}}.
\end{equation}
Here, $t_{\rm original}$ is the $T$-count of the input circuit before unoptimization. In this paper, we use the single multi-Pauli rotation gate $R_{ZZ\dots Z}\qty(\pi/4)$ as the input circuit, which has a theoretical $T$-count of 1, i.e., $t_{\rm original} = 1$.
A larger value of $t_{\rm unopt}$ indicates a higher degree of unoptimization. The greater the difference between $t_{\rm unopt}$ and $t_{\rm opt}$, the better the performance of the compiler, which corresponds to a higher value of $p$.
Since $p$ is normalized by the maximum possible reduction in $T$-count, it always falls within the range of $0$ to $1$. In particular, $p=1$ indicates perfect optimization, where the compiler completely restores the original circuit.

\begin{figure}[t]
  \begin{algorithm}[H]
    \caption{Compiler benchmark}
    \begin{algorithmic}[1]
      \State \textbf{Input:} $U$ $n$-qubit circuit
      \State Obtain an unoptimized circuit $V$ using \textbf{Algorithm 4}
      \State Decompose $U$ into a Clifford+$T$~gate set
      \State Obtain $T$-count: $t_{\text{original}}$
      \State Decompose $V$ into a Clifford+$T$~gate set
      \State Obtain $T$-count: $t_{\text{unopt}}$
      \State Apply $T$~gate optimization compiler to $V$
      \State Obtain optimized $T$-count: $t_{\text{opt}}$
      \State Compute gate reduction rate: $p$
      \State \quad $p \gets (t_{\text{unopt}} - t_{\text{opt}}) / (t_{\text{unopt}} - t_{\text{original}})$
      \State \textbf{Output:} $t_{\text{unopt}}, t_{\text{opt}}, p$
    \end{algorithmic}
    \label{alg: benchmark}
  \end{algorithm}
\end{figure}
\begin{figure}[t]
  \centering
  \includegraphics[width=1.0\linewidth]{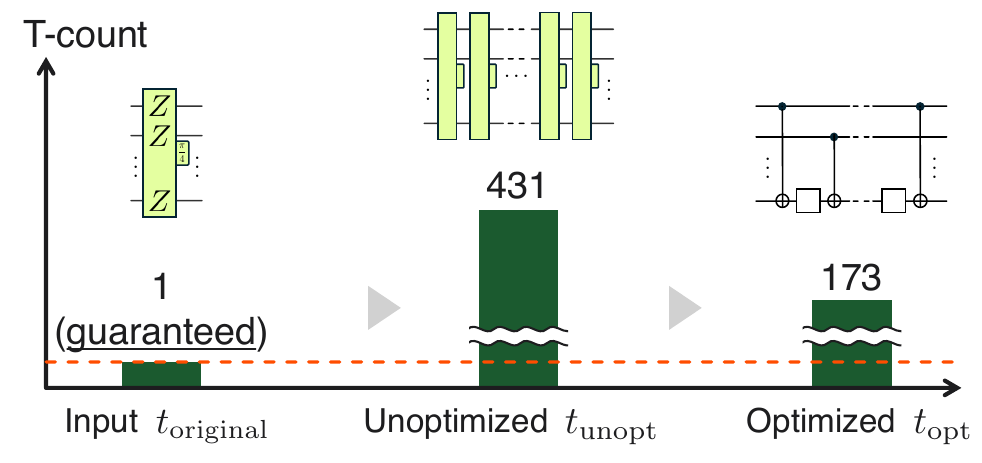}
  \caption{An example workflow for assessing quantum compiler performance using an unoptimized benchmark circuit.
    The input circuit is the rotation $R_{ZZ\dots Z}\qty(\pi/4)$, whose theoretical $T$-count is 1 (indicated by the red dashed line).
    In this illustrative case, applying Algorithm~\ref{alg: gen_unopt} in Appendix~\ref{sec: mcr_unopt} produces an unoptimized Clifford$+T$ circuit $V$ with $t_{\rm unopt} = 431$.
    We then decompose $V$ into a Clifford$+T$ gate set and pass the resulting circuit to a compiler. This yields an optimized circuit with $t_{\rm opt} = 173$, corresponding to a 60\% $T$-count reduction as defined in Eq.~\eqref{eq:reduction_rate}.}
  \label{fig:benchmark_example}
\end{figure}

We use the following four state-of-the-art $T$-gate optimization compilers in our experiments:
\begin{enumerate}
  \item Pytket~\cite{pytket}\\
        Developed by Quantinuum, Pytket is a general-purpose compiler that supports optimization for NISQ devices. Although its primary focus is on circuit depth reduction, it also includes a $T$-count optimization pass. In our experiments, we use \verb|RemoveRedundancies| to eliminate unnecessary $T$~gates.
  \item PyZX~\cite{pyzx}\\
        PyZX applies optimization techniques based on ZX-calculus, which represents a quantum circuit as a graphical diagram. It includes graph-theoretic optimization techniques such as local complementation and pivoting~\cite{graph_theoretic_zx_culculus,phase_gadget_zx_calculus} using $X$-Spiders and $Z$-Spiders. We use \verb|full_reduce| function for maximum simplification in the ZX-diagram, and then convert it back to a quantum circuit.
  \item TMerge~\cite{zhang_algorithm}\\
        TMerge reduces the $T$-count by exploiting the commutativity of Pauli rotation axes. As described in Sec.~\ref{subsec: sep_circuits}, it first separates the input Clifford+$T$ circuit into a Clifford gate sequence and a sequential PBC. This compiler can reorder them within each $T$ layer and merge rotation gates that have the same axis.
  \item FastTODD~\cite{fasttodd}\\
        FastTODD uses phase-polynomial as an IR and optimizes the $T$-count by solving the symmetric tensor rank decomposition problem~\cite{tensor_rank}. It removes all internal Hadamard gates by introducing ancilla qubits and applies a fast version of the Third-Order Duplicate and Destroy (TODD) algorithm~\cite{campbell_compiler} to minimize the $T$-count.
\end{enumerate}

\begin{table*}[t]

  \caption{Comparison of $T$~gate optimization results across five compilers. $T$-counts before ($t_{\rm unopt}$) and after ($t_{\rm opt}$) optimization are shown for quantum circuits with $n$ qubits. The percentage reduction in $T$-count is denoted as Red.~(\%), defined by Eq.~\eqref{eq:reduction_rate}. For \textbf{FastTODD}, $n_h$ indicates the number of ancilla qubits introduced during Hadamard gadgetization. We report the optimization results averaged over 100 samples.}
  \setlength{\tabcolsep}{5pt}
  \centering
  \begin{tabular}{llrrrrrrrrrrr}
    \toprule
    \multirow{2}{*}[-1.07em]{$n$} & \multirow{2}{*}[-1.07em]{$t_{\rm unopt}$} & \multicolumn{2}{c}{\textbf{Pytket}} & \multicolumn{2}{c}{\textbf{PyZX}} & \multicolumn{2}{c}{\textbf{TMerge}} & \multicolumn{3}{c}{\textbf{FastTODD}} & \multicolumn{2}{c}{\textbf{MCR Compiler}}                                                                                    \\
    \cmidrule(lr){3-4} \cmidrule(lr){5-6} \cmidrule(lr){7-8} \cmidrule(lr){9-11} \cmidrule(lr){12-13}
                                  &                                           & $t_{\rm opt}$                       & Red. (\%)                         & $t_{\rm opt}$                       & Red. (\%)                             & $t_{\rm opt}$                             & Red. (\%) & $t_{\rm opt}$ & $n_h$  & Red. (\%) & $t_{\rm opt}$ & Red. (\%)       \\
    \midrule                      
    2                             & 46.78                                     & 43.44                               & 7.28                              & 32.68                               & 30.77                                 & 32.68                                     & 30.77     & 32.56         & 18.72  & 31.03     & 1.00          & \textbf{100.00} \\
    3                             & 106.34                                    & 104.26                              & 1.98                              & 94.84                               & 10.92                                 & 94.84                                     & 10.92     & 93.45         & 57.77  & 12.24     & 1.10          & \textbf{99.91}  \\
    4                             & 190.34                                    & 189.74                              & 0.32                              & 185.50                              & 2.56                                  & 185.50                                    & 2.56      & 185.20        & 115.80 & 2.71      & 1.00          & \textbf{100.00} \\
    5                             & 298.14                                    & 297.92                              & 0.074                             & 295.96                              & 0.73                                  & 295.96                                    & 0.73      & 295.76        & 189.30 & 0.80      & 1.16          & \textbf{99.95}  \\
    6                             & 430.38                                    & 430.28                              & 0.023                             & 429.74                              & 0.15                                  & 429.74                                    & 0.15      & 429.66        & 278.50 & 0.17      & 1.20          & \textbf{99.95}  \\
    7                             & 586.30                                    & 586.30                              & 0.00                              & 586.06                              & 0.041                                 & 586.06                                    & 0.041     & 586.00        & 381.96 & 0.051     & 1.14          & \textbf{99.98}  \\
    8                             & 766.12                                    & 766.12                              & 0.00                              & 766.12                              & 0.00                                  & 766.12                                    & 0.00      & 766.12        & 503.12 & 0.00      & 1.80          & \textbf{99.90}  \\
    9                             & 969.90                                    & 969.86                              & 0.0041                            & 969.74                              & 0.017                                 & 969.74                                    & 0.017     & 969.74        & 637.69 & 0.017     & 1.70          & \textbf{99.93}  \\
    \bottomrule
  \end{tabular}
  \label{tab:compiler_performance}
\end{table*}

\begin{table*}[t]
  \caption{Comparison of $T$~gate optimization results across five compilers without MCR-swap.
    The same definitions for $t_{\rm unopt}$, $t_{\rm opt}$, and $n_h$ apply as in Table~\ref{tab:compiler_performance}.
  }
  \setlength{\tabcolsep}{5pt}
  \centering
  \begin{tabular}{llrrrrrrrrrrr}
    \toprule
    \multirow{2}{*}[-1.07em]{$n$} & \multirow{2}{*}[-1.07em]{$t_{\rm unopt}$} & \multicolumn{2}{c}{\textbf{Pytket}} & \multicolumn{2}{c}{\textbf{PyZX}} & \multicolumn{2}{c}{\textbf{TMerge}} & \multicolumn{3}{c}{\textbf{FastTODD}} & \multicolumn{2}{c}{\textbf{MCR Compiler}}                                                                                    \\
    \cmidrule(lr){3-4} \cmidrule(lr){5-6} \cmidrule(lr){7-8} \cmidrule(lr){9-11} \cmidrule(lr){12-13}
                                  &                                           & $t_{\rm opt}$                       & Red. (\%)                         & $t_{\rm opt}$                       & Red. (\%)                             & $t_{\rm opt}$                             & Red. (\%) & $t_{\rm opt}$ & $n_h$  & Red. (\%) & $t_{\rm opt}$ & Red. (\%)       \\
    \midrule                      
    2                             & 33.00                                     & 33.00                               & 0.00                              & 33.00                               & 0.00                                  & 33.00                                     & 0.00      & 33.00         & 18.96  & 0.00      & 1.00          & \textbf{100.00} \\
    3                             & 73.00                                     & 73.00                               & 0.00                              & 73.00                               & 0.00                                  & 73.00                                     & 0.00      & 73.00         & 47.36  & 0.00      & 1.00          & \textbf{100.00} \\
    4                             & 129.00                                    & 129.00                              & 0.00                              & 129.00                              & 0.00                                  & 129.00                                    & 0.00      & 129.00        & 86.96  & 0.00      & 1.00          & \textbf{100.00} \\
    5                             & 201.00                                    & 201.00                              & 0.00                              & 201.00                              & 0.00                                  & 201.00                                    & 0.00      & 201.00        & 138.51 & 0.00      & 1.00          & \textbf{100.00} \\
    6                             & 289.00                                    & 289.00                              & 0.00                              & 289.00                              & 0.00                                  & 289.00                                    & 0.00      & 289.00        & 202.30 & 0.00      & 1.00          & \textbf{100.00} \\
    7                             & 393.00                                    & 393.00                              & 0.00                              & 393.00                              & 0.00                                  & 393.00                                    & 0.00      & 393.00        & 277.37 & 0.00      & 1.00          & \textbf{100.00} \\
    8                             & 513.00                                    & 513.00                              & 0.00                              & 513.00                              & 0.00                                  & 513.00                                    & 0.00      & 513.00        & 363.76 & 0.00      & 1.00          & \textbf{100.00} \\
    9                             & 649.00                                    & 649.00                              & 0.00                              & 649.00                              & 0.00                                  & 649.00                                    & 0.00      & 649.00        & 461.86 & 0.00      & 1.00          & \textbf{100.00} \\
    \bottomrule
  \end{tabular}
  \label{tab:compiler_performance_without_swap}
\end{table*}

\begin{table*}[t]
  \centering
  \caption{$T$-count optimization performance for the kicked XY model. \textbf{Original} denotes the $T$-count of the original quantum circuit. \textbf{MCR Compiler} shows the $T$-count after optimization, in comparison with that obtained by \textbf{TMerge}. \textbf{Ideal} denotes the $T$-count obtained under the assumption that all single-qubit $Z$-rotation gates are synthesized into Clifford gates.}
  \label{tab:kicked_xy}

  \begin{minipage}[t]{0.48\linewidth}
    \centering
    (a) One-dimensional chain (Open boundary condition).

    \vspace{.5mm}
    \begin{tabular}{lrrrr}
      \toprule
      $n$ & \textbf{Original} & \textbf{TMerge} & \textbf{MCR Compiler} & \textbf{Ideal} \\
      \midrule
      6   & 96                & 96              & \textbf{60}           & 60             \\
      12  & 408               & 408             & \textbf{264}          & 264            \\
      18  & 936               & 936             & \textbf{612}          & 612            \\
      24  & 1680              & 1680            & \textbf{1104}         & 1104           \\
      \bottomrule
    \end{tabular}
  \end{minipage}
  \hfill
  \begin{minipage}[t]{0.48\linewidth}
    \centering
    (b) Two-dimensional square lattice.

    \vspace{.5mm}
    \begin{tabular}{lrrrr}
      \toprule
      $n$         & \textbf{Original} & \textbf{TMerge} & \textbf{MCR Compiler} & \textbf{Ideal} \\
      \midrule
      $2\times 2$ & 48                & 48              & \textbf{32}           & 32             \\
      $3\times 3$ & 330               & 330             & \textbf{262}          & 240            \\
      $4\times 4$ & 1024              & 1024            & \textbf{862}          & 768            \\
      $5\times 5$ & 2730              & 2730            & \textbf{2132}         & 2080           \\
      \bottomrule
    \end{tabular}
  \end{minipage}
\end{table*}

We create the benchmark circuits via quantum circuit unoptimization using the quantum circuit simulator Stim~\cite{stim}. We can entirely characterize the type of a multi-Pauli rotation gate by a Pauli string specifying the rotation axis and a sign indicating whether the rotation angle is $\pi/4$ or $-\pi/4$. Stim is well-suited for dealing with such circuits, as it allows efficient updates of Pauli strings under Clifford conjugation while preserving small representations.
Since the unoptimized circuit consists solely of multi-Pauli rotation gates with angles of $\pm \pi/4$, they can be exactly decomposed into Clifford$+T$ gates.
To save benchmark circuit data, we use the Qasm~\cite{qasm} and qc file formats~\cite{reversible_benchmark}.
To ensure that the unoptimization process preserves circuit equivalence as defined in Eq.~\eqref{eq:equivalence}, we verify correctness using the MQT-qcec tool provided by The Munich Quantum Toolkit~\cite{mqt}. All source code used in this work is available on GitHub~\cite{my_github}.

Table~\ref{tab:compiler_performance} summarizes the averaged results.
The table includes the $T$-count before optimization, followed by the output from each compiler: the reduced $T$-count and the corresponding reduction percentage.
For FastTODD, an additional column shows the number of ancilla qubits $n_h$ introduced via Hadamard gadgetization.
As expected, we verify that the $T$-counts $t_{\rm unopt}$ and $t_{\rm opt}$ increase as the number of qubits $n$ increases.

We first assess the performance of the MCR Compiler proposed in Sec.~\ref{subsec: mcr_aware_compiler} on this benchmark dataset. The MCR Compiler consistently achieves $T$-count reduction rates close to $100\%$ across all values of $n$, effectively restoring the $T$-count to that of the original circuit.
This demonstrates that MCR-based transformations can serve as an effective optimization primitive for $T$-count reduction.

In contrast, none of the four existing compilers produces reductions comparable to those obtained by the MCR Compiler at any circuit size.
Among them, FastTODD consistently achieves the largest $T$-count reductions. For instance, at $n=2$, it reduces the $T$-count by over 30\%. However, its effectiveness rapidly drops as $n$ increases, reaching well below 1\% for $n \geq 6$. Moreover, FastTODD introduces a large number of ancilla qubits $n_h$ during Hadamard gadgetization. For example, at $n=9$, it requires over 630 ancilla qubits, increasing the total number of qubits by more than 70 times. Compilation becomes computationally expensive at this scale, with runtimes exceeding 21 h per circuit on average, despite yielding only marginal improvements in $T$-count.
In comparison, PyZX and TMerge provide modest reductions for small circuits, but their performance similarly degrades as $n$ grows, with both reaching reductions below the 1\% level at $n = 5$. Pytket, on the other hand, consistently shows a relatively lower reduction rate compared to other compilers.

To explore why optimization partially works for small circuits but not for larger ones, we conduct an additional numerical experiment using a benchmark dataset generated by skipping step (c) in Fig.~\ref{fig:unopt_workflow}, i.e., omitting gate swapping based on MCR (see Appendix~\ref{sec: mcr_unopt} for details).
We report the results in Table~\ref{tab:compiler_performance_without_swap}, revealing a complete absence of $T$-count reduction across four existing compilers and circuit sizes. Even FastTODD, despite the introduction of ancilla qubits, fails to achieve any improvement under this condition.
This result suggests that for circuits with a small number of qubits, MCR-based gate swapping accidentally matches the rotation axes of neighboring gates, making optimization easier. Since the total number of gate sequences satisfying MCR increases exponentially with the number of qubits $n$ (see Appendix~\ref{sec: mcr_scaling} for details), cases where the rotation axes coincide after swapping become less likely as $n$ increases.

\subsection{Application for practical quantum circuits}
\label{subsec: mcr_application}

As briefly introduced in Sec.~\ref{subsec: definition_mcr}, we consider quantum circuits arising from the simulation of kicked spin dynamics of the XY model, whose Hamiltonian is given in Eq.~\eqref{eq: xy_model}, to assess the practical relevance of MCR-aware optimization.
In this work, all rotation angles are fixed to $\pi/4$, so that each rotation gate corresponds to one $T$~gate. Moreover, the number of layers is always chosen to be even: for a given number of qubits $n$, we use $n$ layers when $n$ is even and  $n+1$ layers otherwise.
Under this setting, MCR-based circuit rewrites allow all $Z$-rotations to be synthesized into Clifford operations, leaving only rotations around the $XX$ and $YY$ axes to contribute to the final $T$-count.

We benchmark these circuits on a one-dimensional chain with open boundary conditions (ranging from $n=~6$ to $24$), and on a two-dimensional square lattice (ranging from $n=4$ to $25$). The optimization results are summarized in Table~\ref{tab:kicked_xy}, where the column Ideal corresponds to the $T$-count obtained when reducible $Z$-rotation gates are successfully merged.
In the one-dimensional case, the MCR Compiler achieves the same $T$-count as the Ideal value, whereas no reduction is obtained using conventional compilers.
For two-dimensional lattices, the MCR Compiler still yields substantial reductions compared to existing methods.
These results indicate that MCR-based optimization is effective not only for deliberately unoptimized benchmark circuits but also for practically relevant quantum circuits derived from realistic quantum algorithms.

\section{Conclusion and outlook}\label{sec: conclusion}

In this paper, we propose and formalize the multiproduct commutation relation (MCR) as a nonlocal and ancilla-free circuit transformation rule for Clifford$+T$ circuits.
MCR enables equivalent rewrites involving multiple nonadjacent multi-Pauli rotation gates, capturing nontrivial commutativity that pairwise commutation rules cannot reach.
Given the rotation axes of three multi-Pauli rotation gates, Theorem~\ref{thm:abc_multi_commutation} provides a concrete condition to determine the fourth gate that satisfies MCR.
Building on this characterization, we develop the MCR Compiler, which combines MCR-based circuit rewrites with local rotation merging to achieve $T$-count reduction.
To assess the optimization capability enabled by MCR, we adopt 
a procedure called quantum circuit unoptimization, which introduces redundancy into Clifford$+T$ circuits by inserting and swapping quantum gates while preserving circuit equivalence. This approach systematically generates complex circuits from simple ones with a guaranteed optimal $T$-count, offering a quantitative framework for evaluating how well a compiler eliminates redundant $T$~gates.

Through numerical experiments within this framework, we demonstrate that the MCR Compiler achieves substantial $T$-count elimination, whereas state-of-the-art $T$-count optimization compilers yield marginal reductions.
We also compare the optimization performance results with and without MCR-based gate swapping. The comparison shows that while MCR-based gate swapping introduces additional circuit redundancy, it can occasionally trigger effective local optimization in small regions of the circuit.
However, such cases become infrequent as the number of qubits grows, due to the exponential increase in the number of valid MCR candidates.

MCR involves circuit transformations that can increase or decrease the number of $T$ layers, thereby suggesting a broader class of rewrite rules beyond the reach of existing compiler strategies. This insight opens opportunities for compiler design, circuit synthesis, and algorithm-based optimization.
However, global exploration of equivalent circuit representations becomes computationally infeasible, since the search space induced by MCR-based rewrites grows exponentially with the system size. The MCR Compiler proposed in this work maintains scalability by restricting the search to a subspace defined by the partition of the circuit into $T$~layers. MCR candidates are examined only across neighboring layers, allowing the optimization procedure to remain tractable even for large circuits.
A natural direction for future work is therefore to extend the exploration beyond this restricted subspace while preserving scalability. Promising approaches include Monte Carlo-style strategies that probabilistically insert identity circuits to expose additional MCR opportunities, as well as divide-and-conquer methods that exploit structural patterns of the target circuits.
Once developed, the performance of such a compiler can be evaluated on representative datasets, such as the \textit{Reversible Benchmarks}~\cite{reversible_benchmark}.

Another promising direction involves generalizing MCR to Clifford$+R_Z$ circuits with continuous rotation angles.
This circuit class is widely used in quantum algorithms, including variational quantum algorithms~\cite{vqe, qaoa} and Trotterized quantum simulation tasks~\cite{quantum_utility}.
In architectures that support transversal Clifford gates and implement $Z$-rotations using repeat-until-success protocols, minimizing the number of such gates can significantly improve execution efficiency~\cite{early_ftqc, toshio_early_ftqc}. Consequently, many quantum algorithms can benefit from an MCR-based optimization tool designed to reduce such rotations.
The authors in Ref.~\cite{optimal_parametric_qc} propose a method to optimize the number of parameters in nonrepeated parameterized quantum circuits. However, this excludes circuits where parameters exhibit algebraic relationships, such as the simultaneous appearance of $\alpha_1$ and $-\alpha_1 + \pi/4$. In contrast, MCR admits simplifications even in these correlated cases, offering opportunities to optimize beyond the reach of current techniques. This property may provide value in structured algorithms, such as QAOA-MC~\cite{qaoa_mc}, where repeated parameters are present.

For circuits using discrete angles, also known as Clifford$+R_Z\qty(\pi/2^k)$ circuits, which relate closely to the Clifford hierarchy~\cite{clifford_hierarchy_gottesman, clifford_hierarchy_anderson}, MCR could offer similar benefits. Prior works~\cite{amy_reed_muller, fasttodd} explore optimization in this setting using Reed-Muller code techniques and prove that this problem can be reduced to the higher-order symmetric tensor rank decomposition problem. The complexity of the symmetric tensor rank decomposition problem remains unknown. However, the related problem of determining the tensor rank of a third-order tensor, without the symmetry constraint, is NP-complete~\cite{tensor_rank}.
We expect MCR-based strategies to improve the compilation performance even in such a circuit structure.
A practical application in this circuit setting is time evolution by probabilistic angle interpolation (TE-PAI)~\cite{tepai}, which simulates quantum dynamics using ensembles of random circuits composed only of multi-Pauli rotations with angles $\pm\Delta$ or $\pm\pi$. Since these circuits belong to the Clifford+$R_Z$ (discrete angle) class, a compiler that reduces redundant rotations could help efficiently implement TE-PAI on a real quantum device.

\vskip\baselineskip

\section*{Acknowledgement}
This work is supported by the MEXT Quantum Leap Flagship Program (MEXT Q-LEAP) Grant No.~JPMXS0120319794, JST COI-NEXT Grant No.~JPMJPF2014, JST Moonshot R\&D Grant No.~JPMJMS2061, and JSPS KAKENHI Grant Nos.~JP25KJ1712 and JP24K16979.

\appendix

\section{Proof of Theorem~\ref{thm:abc_multi_commutation}}
\label{sec: proof_of_mcr}

We prove Theorem~\ref{thm:abc_multi_commutation} as described in Sec.~\ref{subsec: mcr_construction_condition}.

\noindent
($\Leftarrow$) This can be demonstrated by substituting $D=-ABC$ into (i)--(iii) of (I) in order.

\noindent
($\Rightarrow$):
Combining (ii) and (iii) of (I) and the given assumptions, we obtain
\begin{align}
   & \quad \commutator{A+B}{C+D} = 0 \\
  \label{eq:proof_multi_commutation}
   & \Rightarrow AC+AD+BC+BD = 0
\end{align}
Since $A, B, C$, and $D$ are mutually distinct Pauli operators, each term in Eq.~\eqref{eq:proof_multi_commutation} is a nonidentity Pauli product. Because the set of $n$-qubit Pauli operators forms an orthogonal basis under the Hilbert–Schmidt inner product, their sum can be zero only if terms cancel in pairs.

\vspace{1em}
\noindent
\underline{Case 1}\quad $AC + BC = 0$

\noindent
Eq.~\eqref{eq:proof_multi_commutation} becomes
\begin{equation}
  \begin{split}
     & \quad AD + BD = 0   \\
     & \Rightarrow A = -B,
  \end{split}
\end{equation}
which implies that $A$ and $B$ have the same rotation axis but opposite rotation angles, contradicting the assumption that they have distinct axes.

\vspace{1em}
\noindent
\underline{Case 2}\quad $AC + BC \neq 0$

\noindent
We must partition the sum $-(AC + BC)$ into two Pauli operators $AD$ and $BD$ such that their sum equals this value. Because a sum of two distinct Pauli operators cannot be another single Pauli operator, the only possibility is that
$\qty(AD, BD) = (-AC, -BC)$ or $\qty(AD, BD) = (-BC, -AC)$.
We analyze these two cases separately.
\begin{itemize}
  \item If $\qty(AD, BD) = (-AC, -BC)$, then $D = -C$, implying that $D$ and $C$ share the same rotation axis, which again contradicts the assumption.
  \item If $\qty(AD, BD) = (-BC, -AC)$, then $D = -ABC (= -BAC)$, which has a rotation axis distinct from all of $A, B$, and $C$, and hence satisfies the assumption.
\end{itemize}
Therefore, the only valid solution is $D = -ABC$.

\section{Analysis of the number of possible MCR pairs}
\label{sec: mcr_scaling}
Here, we explain the total number of combinations of gate rotation axes that satisfy MCR.
MCR is a commutation rule applied to four sequential gates: two on the left and two on the right.
Let the set of rotation axes be $\mP_n^* \coloneqq \qty{ \pm 1 } \times \qty{I, X, Y, Z}^{\otimes n} \setminus \qty{\pm I^{\otimes n}}$ (Eq.~\eqref{eq:rotation_axis_set}),
and let the rotation axes of the gates from the left be $A, B, C, D \in \mP_n^*$, where each gate has a different rotation axis except for the sign $\pm 1$.
According to Algorithm~\ref{alg: gen_unopt} in Sec.~\ref{subsec: unopt_recipe}, the method for choosing rotation axes that satisfy MCR is as follows.

\begin{enumerate}
  \item Select one element from $\mP_n^*$ and set it as $A$.
  \item Choose $B \in \mP_n^*$ such that $B$ commutes with $A$.
  \item Choose $C \in \mP_n^*$ such that $\anticommutator{A}{C} = \anticommutator{B}{C} =0$.
  \item Obtain $D$ by setting $D=-ABC$.
\end{enumerate}
Let us determine the number of ways to select the rotation axes according to this procedure.
First, the axis $A$ can be chosen from any nonidentity Pauli operator, giving $4^n-1$ possible choices. The second axis $B$ must commute with $A$ and must not be equal to $A$ or the identity. In general, for any given Pauli operator, half of the other operators commute with it. Therefore, the number of valid choices for $B$ is
$4^n /2 - 2$, where the subtraction accounts for excluding the identity operator and the previously chosen $A$.
Next, $C$ must anticommute with both $A$ and $B$. Since one-fourth of all Pauli operators anticommute with both, the number of valid choices for $C$ is $4^{n-1}$.
$D$ is determined uniquely as $D=-ABC$ to satisfy MCR once $A$, $B$, and $C$ are fixed as described in Theorem~\ref{thm:abc_multi_commutation}.
Additionally, since the signs of $A$, $B$, and $C$ can be chosen independently, we multiply by 8 to account for the possible sign combinations. However, to avoid overcounting due to permutation of the pairs $\qty(A, B)$ and $\qty(C, D)$, and the exchange between these two pairs, we divide the total count by $(2!)^3=8$. Putting everything together, the total number of such combinations is
\begin{equation}
  \frac{1}{8} \cdot 4^n \qty(4^n-1) \qty(4^n-4),
  \label{eq:mcr_candidate}
\end{equation}
which shows that the number of possible axis configurations satisfying MCR scales exponentially with $n$.

\section{Construction of MCR-based unoptimization benchmark circuits}
\label{sec: mcr_unopt}
\subsection{MCR-based recipes for unoptimization}
\label{subsec: unopt_recipe}

\begin{figure}[t]
  \centering
  \includegraphics[width=.9\linewidth]{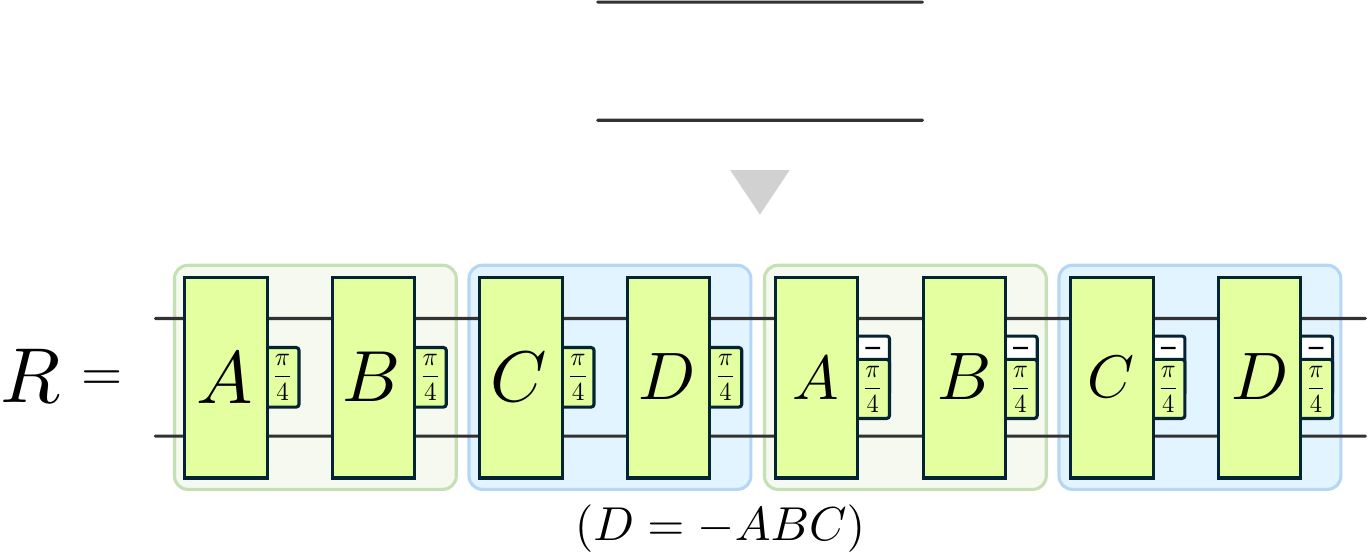}
  \caption{Construction of a nontrivial identity operator based on the multiproduct commutation relation property, used in the insertion step of the unoptimization procedure. Each alphabet represents a rotation axis of a multi-Pauli rotation.}
  \label{fig:nontrivial_identity_circuit}
\end{figure}

\begin{figure}[t]
  \begin{center}
    \includegraphics[width=\linewidth]{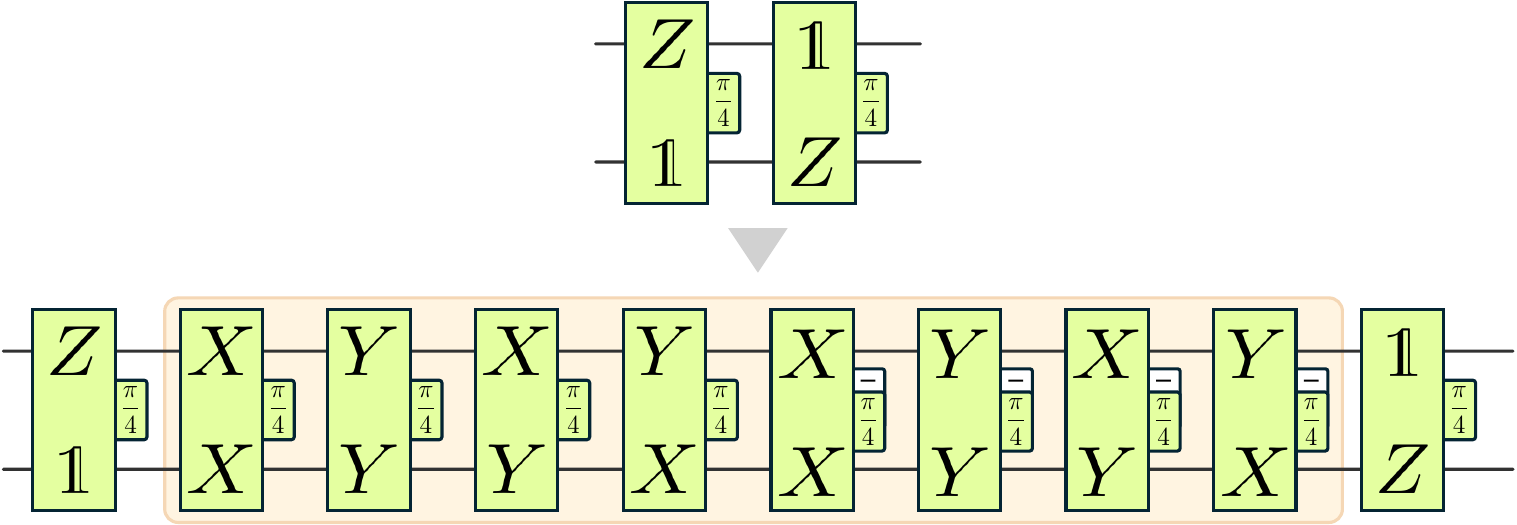}
    \caption{A gate insertion example based on Fig~\ref{fig:nontrivial_identity_circuit}. Although the inserted circuit is equivalent to the identity, it cannot be simplified by considering only pairwise commutativity of adjacent gates.}
    \label{fig:rot_insertion_nontrivial}
  \end{center}
\end{figure}

\begin{figure}[ht]
  \begin{algorithm}[H]
    \caption{MCR-based quantum circuit unoptimization}
    \begin{algorithmic}[1]
      \State \textbf{Generation of identity operator using MCR}
      \State Select $A \in \mP_n^*$
      \State Choose $B \in \mP_n^*$ such that $[A, B] = 0$
      \State Choose $C \in \mP_n^*$ such that $\{A, C\} = 0$ and $\{B, C\} = 0$
      \State Compute $D \gets -A B C$
      \State Construct gate sequence:\\
      \quad $R \gets R_{-D}(\pi/4) R_{-C}(\pi/4) R_{-B}(\pi/4) R_{-A}(\pi/4)$\\
      \quad $R_{D}(\pi/4) R_{C}(\pi/4) R_{B}(\pi/4) R_{A}(\pi/4)$
      \State \textbf{Output:} $R$ Gate sequence such that $R  = I$

      \State
      \State \textbf{Generation of unoptimized circuit}
      \For{$i \gets 0$ to $n^2-1$}
      \State $C[i] \gets$ $R$, \verb|Gate Swapping|
      \EndFor
      \State $V \gets \mO(U, C)$
    \end{algorithmic}
    \label{alg: gen_unopt}
  \end{algorithm}
\end{figure}

\begin{figure}[t]
  \begin{center}
    \includegraphics[width=.7\linewidth]{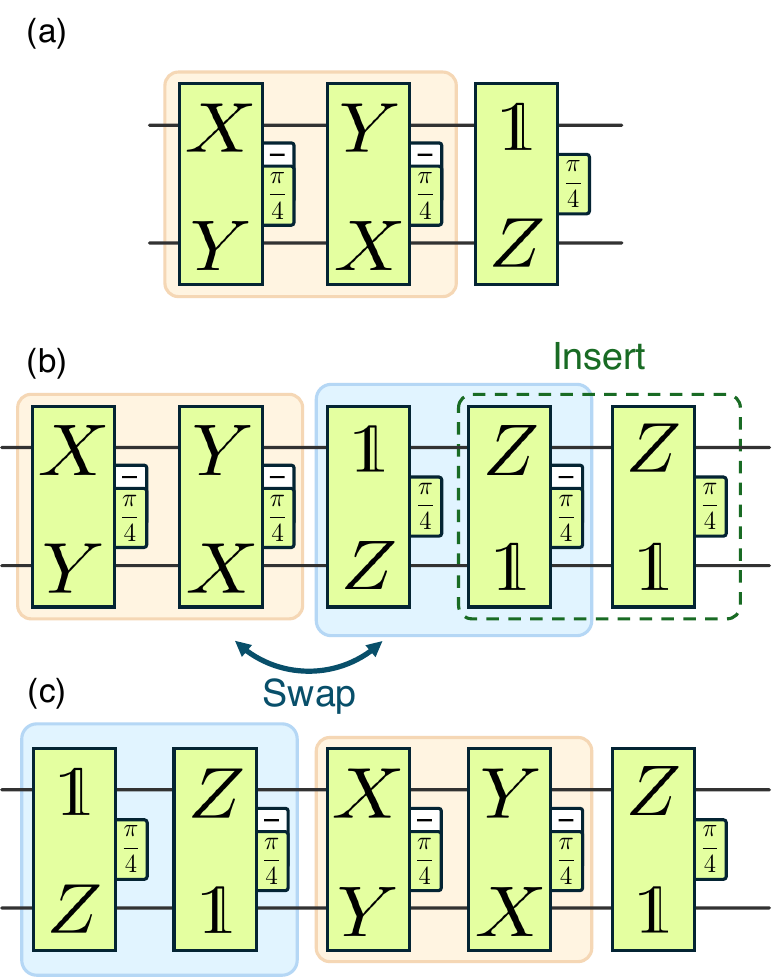}
    \caption{Gate swapping unoptimization based on multiproduct commutation relation. \\
      (a) The circuit segment after the gate insertion.
      (b) Additional gate insertion of identity gates (green dashed box) to satisfy multiproduct commutation relation.
      (c) The resulting circuit after swapping the commuting blocks.}
    \label{fig:rot_gate_swap_nontrivial}
  \end{center}
\end{figure}

First, we explain \textit{gate insertion}, which refers to the operation of inserting a gate $A$ together with its inverse $A^{\dagger}$ placed adjacently in the circuit, yielding the identity operator.
A typical example is the insertion of $R_P\qty(\theta)$ followed by $R_P\qty(-\theta)$.
In this paper, we generalize this concept by employing MCR to construct more complex identity sequences, denoted $R$, as depicted in Fig.~\ref{fig:nontrivial_identity_circuit}.
Such sequences are systematically generated by selecting rotation axes satisfying the MCR conditions and arranging their gates and inverses, as formalized in Theorem~\ref{thm:abc_multi_commutation}.
An example is shown in Fig.~\ref{fig:rot_insertion_nontrivial}, and the generation procedure is given in Algorithm~\ref{alg: gen_unopt}.
\begin{figure*}[t]
  \centering
  \includegraphics[width=1.0\linewidth]{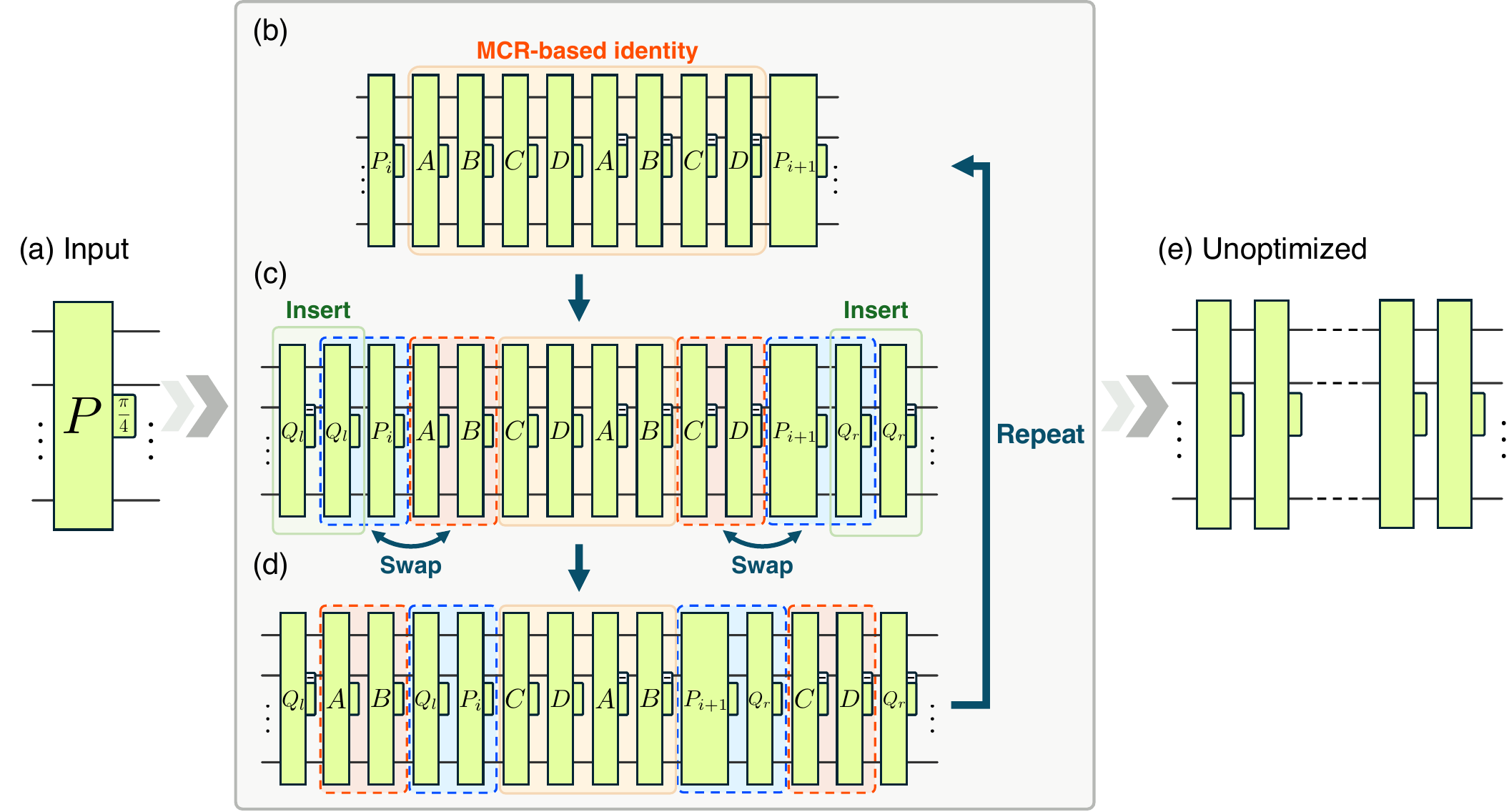}
  \caption{Overview of multiproduct commutation relation-based quantum circuit unoptimization.
    (a) The input circuit, consisting of a single multi-Pauli rotation gate.
    (b) Insertion of a multiproduct commutation relation-based identity sequence (as constructed in Fig.~\ref{fig:nontrivial_identity_circuit}).
    (c) Insertion of additional gates $R_{Q_i}^{\dagger}$ and $R_{Q_i}$ at both ends, corresponding to the procedure described in Fig.~\ref{fig:rot_gate_swap_nontrivial}(b).
    (d) The circuit after swapping the outermost gates, thereby completing the procedure in Fig.~\ref{fig:rot_gate_swap_nontrivial}.
    (e) The unoptimized circuit obtained by repeating steps (b) -- (d) $n^2$ times.}
  \label{fig:unopt_workflow}
\end{figure*}

\begin{figure}[h]
  \centering
  \includegraphics[width=.8\linewidth]{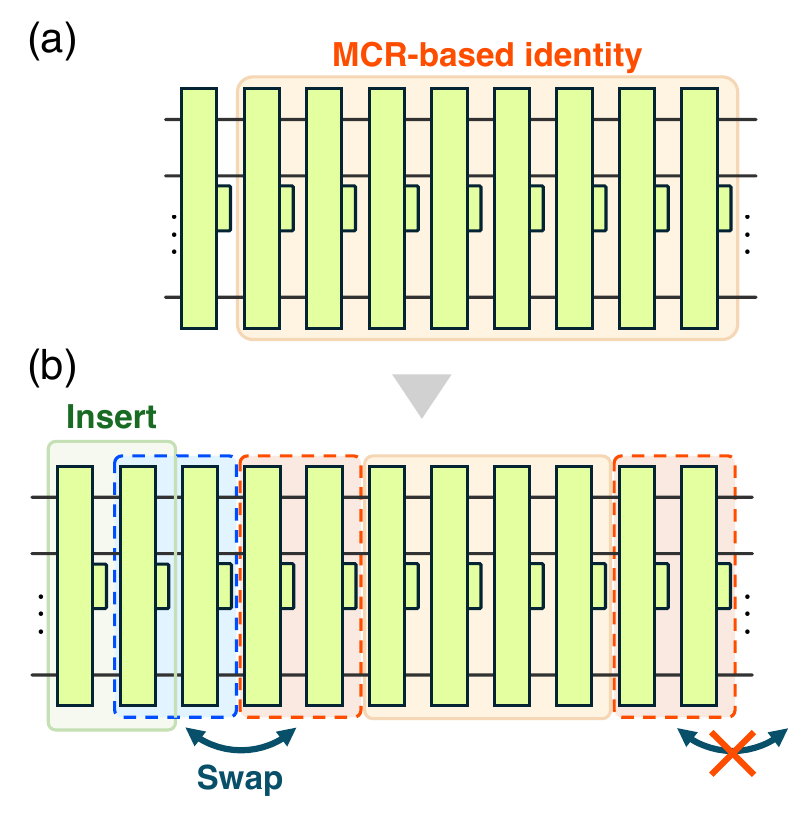}
  \caption{Special case of gate swapping in multiproduct commutation relation-based unoptimization when the target gate is located at the rightmost end of the circuit. (a) Circuit after inserting a multiproduct commutation relation-based identity operator. (b) Swap operation applied only to the left end, as the right end cannot be operated in this configuration.}
  \label{fig:edge_selection}
\end{figure}

Next, we explain the operation of \textit{gate swapping}, which refers to exchanging the positions of two adjacent gate blocks in a circuit.
Here, we combine the gate swapping with an additional gate insertion that induces commutativity through MCR. This procedure is illustrated in Fig.~\ref{fig:rot_gate_swap_nontrivial}.
Consider the swap of the rightmost two gates in the 8-gate identity sequence constructed via MCR (as shown in Fig.~\ref{fig:rot_insertion_nontrivial}) with the gate immediately to their right.
Since MCR can only be applied to pairs of gates, it is necessary to add one more gate to complete the MCR condition. This supplementary gate can be computed using Theorem~\ref{thm:abc_multi_commutation}; in this case, $-ABC =-Z I$.
Therefore, by adding $R_{ZI}\qty(-\pi/4)$, a gate sequence satisfying MCR can be constructed.
Thus, as shown in Fig.~\ref{fig:rot_gate_swap_nontrivial}(b), if we insert $R_{ZI}\qty(-\pi/4)$ and $R_{ZI}\qty(\pi/4)$ to the right of $R_{YX}\qty(-\pi/4)$, we can swap the two pairs of gates $(R_{XY}\qty(-\pi/4), R_{YX}\qty(-\pi/4))$ and $(R_{IZ}\qty(\pi/4), R_{ZI}\qty(-\pi/4))$ since they satisfy MCR. The resulting reordered sequence is shown in Fig.~\ref{fig:rot_gate_swap_nontrivial}(c).

\subsection{Procedure for benchmark circuit construction}
\label{subsec: benchmark_generation}
We explain the unoptimization algorithm that combines the gate insertion and gate swapping techniques described above, along with the compiler benchmarking methodology.
The overall workflow is illustrated in Fig.~\ref{fig:unopt_workflow}.
We begin with an $n$-qubit circuit $U$ consisting of a single multi-Pauli rotation gate $R_P\qty(\pi/4)$, which guarantees an optimal $T$-count of 1 and thus cannot be further reduced. Then, from the circuit, we randomly select one multi-Pauli rotation gate with axis $P_i$ and insert an identity operator based on Algorithm~\ref{alg: gen_unopt} immediately to its right. We make the identity sequence such that the following conditions are satisfied:
\begin{equation}
  \label{eq:p_i_anticommutation}
  \anticommutator{P_i}{A} = \anticommutator{P_i}{B} = \anticommutator{C}{P_{i+1}}= \anticommutator{D}{P_{i+1}} = 0.
\end{equation}
Next, we perform gate swapping on the two gates at both ends of the inserted gate sequence. Here, we first explain the operation of the gates in the left half. In Fig.~\ref{fig:unopt_workflow}(c), we insert an additional gates $R_{Q_l}\qty(-\pi/4)$ and $R_{Q_l}\qty(\pi/4)$ (highlighted in green), where $Q_l = -P_i AB$. Since the gate pairs $(R_{Q_l}, R_{P_i})$ (highlighted in blue) and $(R_A, R_B)$ (highlighted in red) satisfy MCR, we can swap these pairs. This operation results in the updated gate sequence depicted in Fig.~\ref{fig:unopt_workflow}(d). We perform a similar operation on the gate sequence in the right half. In this case, $Q_r = -CD P_{i+1}$, which also satisfies MCR with the adjacent pairs $(R_{-C}, R_{-D})$ and $(R_{P_{i+1}},R_{Q_r})$.
Note that when we select the gate with axis $P_i$ as the right edge of the circuit, we swap the gates on only the left side as shown in Fig.~\ref{fig:edge_selection}. In this case, we do not add the additional gates $R_{Q_r}\qty(\pi/4), R_{Q_r}\qty(-\pi/4)$, and only the left-side gates are swapped with the inserted gate sequence.

We repeat these procedures $n^2$ times, where $n$ is the number of qubits in the circuit. The final circuit obtained is the unoptimized circuit $V$, as shown in Fig.~\ref{fig:unopt_workflow}(e).
The gate insertion and swapping methods enable the systematic transformation of a simple input circuit into a structurally distinct one by changing the type, position, and order of the inserted gates.
In our implementation, identity sequences are randomly generated using Algorithm~\ref{alg: gen_unopt}. If the resulting sequence does not satisfy the anticommutation conditions required for MCR-based gate swapping (Eq.~\eqref{eq:p_i_anticommutation}), the generation step is repeated until a valid candidate is found.

\bibliography{ref.bib}

\end{document}